\documentclass[a4paper,12pt]{article}
\usepackage[utf8]{inputenc}
\usepackage[T2A]{fontenc}
\usepackage{indentfirst}
\usepackage{amsmath}
\usepackage{amsthm}
\usepackage{amsfonts}
\usepackage{amssymb}
\usepackage{graphicx}
\usepackage{setspace}
\usepackage{array}
\usepackage{float}
\usepackage{url}
\usepackage{multirow}
\usepackage{multicol}
\usepackage{enumitem}
\usepackage[title]{appendix}
\usepackage{titling}
\usepackage{longtable}

\usepackage{booktabs}
\usepackage{comment}
\usepackage{mathtools}
\usepackage{hyperref}

\usepackage[left=2cm, right=2cm, top=2cm, bottom=2cm]{geometry}

\newtheorem{theorem}{Theorem}
\newtheorem{proposition}{Proposition}
\newtheorem{lemma}{Lemma}
\newtheorem{corollary}{Corollary}

\theoremstyle{remark}%
\newtheorem{example}{Example}%
\newtheorem{remark}{Remark}%

\newcommand{\adprxsym}{\ensuremath{\mathrm{adp}^{\mathrm{RX}}}}
\newcommand{\adprx}[4]{\ensuremath{\adprxsym (#1, #2 \stackrel{#4}{\to} #3)}}
\newcommand{\adpgeneral}[1]{\mathrm{adp}^{#1}}
\newcommand{\adpsym}{\ensuremath{\mathrm{adp}^{\oplus}}}
\newcommand{\adp}[3]{\ensuremath{\adpsym (#1, #2 \to #3)}}

\newcommand{\adps}[4]{\mathrm{adp}^{#1}(#2, #3 \to #4)}
\newcommand{\adpr}[5]{\mathrm{adp}^{#1}(#2, #3 \stackrel{#5}{\to} #4)}
\newcommand{\adpxrsym}[0]{\mathrm{adp}^{\mathrm{XR}}}
\newcommand{\adpxr}[4]{\adpxrsym(#1, #2 \stackrel{#4}{\to} #3)}
\newcommand{\cadppar}[1]{\mathrm{adp}^{\oplus}_{#1}}
\newcommand{\padppar}[2]{\mathrm{adp}^{\oplus}_{#1, #2}}
\newcommand{\cadp}[4]{\cadppar{#1}(#2, #3 \to #4)}
\newcommand{\padp}[5]{\padppar{#1}{#2}(#3, #4 \to #5)}

\newcommand{\wt}[1]{\mathrm{wt}(#1)}
\newcommand{\inv}[2]{#1^{[#2]}}
\newcommand{\cctn}[2]{(#1, #2)}
\newcommand{\zspc}[1]{\ensuremath{\mathbb{Z}_{2}^{#1}}}
\newcommand{\nzer}[2]{\mathcal{N}_{#1}^{#2}}

\usepackage{array}

\allowdisplaybreaks

\providecommand{\keywords}[1]{
	\medskip
    \noindent
	\small\textbf{Keywords: } #1
	}

\providecommand{\affili}[1]{
	\begin{center} 
	\small #1
	\end{center}
	\vspace{.1cm}}

\title{On additive differential probabilities of the composition of bitwise exclusive-or and a bit rotation}
\author{\normalfont Nikolay Kolomeec, Ivan Sutormin, Denis Bykov,\\ Matvey Panferov, Tatyana Bonich
\medskip
}
\date{}

\begin{document}
\large
%
%
\maketitle
%
%
\par\vspace{-60pt}
%
%
 \affili{
	Novosibirsk State University, Novosibirsk, Russia\\

    \medskip

 	\texttt{ kolomeec@math.nsc.ru, ivan.sutormin@gmail.com, den.bykov.2000i@gmail.com, magold1102@gmail.com, tanya50997@gmail.com 
 	}
}
%
%
\begin{abstract}
Properties of the additive differential probability $\adpxrsym$ of the composition of bitwise XOR and a bit rotation are investigated, where the differences are expressed using addition modulo $2^n$. This composition is widely used in ARX constructions consisting of additions modulo $2^n$, bit rotations and bitwise XORs. Differential cryptanalysis of such primitives may involve maximums of $\adpxrsym$, where some of its input or output differences are fixed. 
Although there is an efficient way to calculate this probability (Velichkov et al, 2011), many of its properties are still unknown.
In this work, we find maximums of $\adpxrsym$, where the rotation is one bit left/right and one of its input differences is fixed. 
Some symmetries of $\adpxrsym$ are obtained as well. We provide all its impossible differentials in terms of regular expression patterns and estimate the number of them. This number turns out to be maximal for the one bit left rotation and noticeably less than the number of impossible differentials of bitwise XOR.


\end{abstract}

\section{Introduction}
ARX is one of the modern architectures for symmetric cryptography primitives that uses only three operations: addition modulo $2^n$ (Addition, $\boxplus$), circular shift (Rotation, $\lll$) and bitwise addition modulo 2 (XOR, $\oplus$). 
Examples of such schemes include block ciphers 
FEAL~\cite{FEAL}, Speck~\cite{SimonSpeck}, 
stream ciphers Salsa20~\cite{Salsa20} and ChaCha~\cite{ChaCha}, 
SHA-3 finalists BLAKE~\cite{BLAKE} and Skein~\cite{Threefish-Skein}, MAC algorithm Chaskey~\cite{Chaskey}.
ARX constructions have many advantages: fast performance and compactness of its program implementation, resistance to timing attacks. However, it is difficult to determine if they are secure against differential cryptanalysis~\cite{BihSha91}. This method requires studying how input differences transform to the output differences.

In this work we consider differences that are expressed using addition modulo $2^n$. In some cases they are more appropriate than XOR differences, especially if round keys are added modulo $2^n$ and the number of XOR operations is much less than the number of modulo $2^n$ operations, see, for instance,~\cite{BiryukovVelichkov2014}. Some information related to differential cryptanalysis of ARX constructions and choosing differences can be found in~\cite{Leurent2012,Leurent2013}. The additive differential probability $\adpgeneral{f}(\alpha_1, \ldots, \alpha_k \to \alpha_{k+1})$ for a function $f: (\zspc{n})^k  \to \zspc{n}$, where $\alpha_1, \ldots, \alpha_{k+1} \in \zspc{n}$, is defined as $$\Pr_{x_1, \ldots, x_k \in \zspc{n}}[f(x_1 \boxplus \alpha_1, \ldots, x_k \boxplus \alpha_k) = f(x_1, \ldots, x_k) \boxplus \alpha_{k+1} ].$$
Here $\alpha_1, \ldots, \alpha_k$ and $\alpha_{k+1}$ are its input differences and its output difference respectively, $x_1, \ldots, x_k$ are uniformly distributed.
The additive differential probability of a rotation ($\adpgeneral{\lll}$) and XOR ($\adpsym$) were studied in~\cite{Berson1992, Daum2004} and~\cite{LipmaaEtAl2004,MouhaEtAl2010,MouhaEtAl2021}.
The formula for the additive differential probability $\adprxsym$ for the composition $(x \lll r) \oplus y$ was obtained in~\cite{VelichkovEtAl2011}. 
Also, it was pointed out that it is inaccurate to calculate differential probability of the composition by assuming that the inputs of the basic operations are independent (the same is true for XOR differences, see, for instance,~\cite{HongEtAl2003}).
Note that we can add $\boxplus$ operation to a composition and express new differential probabilities from the old ones in a direct way, since the considered differences go through $\boxplus$ with probability one.
For instance, it works for the function $((x \boxplus z) \lll r) \oplus y$ if we know differential probabilities for $(x \lll r) \oplus y$.
 
We investigate the properties of $\adpxrsym$ for the function $(x \oplus y) \lll r$, where $x, y \in \zspc{n}$ and $1 \leq r \leq n - 1$. They are similar to the properties of $\adprxsym$ since $\adpxr{\alpha}{\beta}{\gamma}{r} = \adprx{\gamma}{\beta}{\alpha}{n - r}$. Though there is the formula for $\adprxsym$, many its properties are still unknown. 
First of all, we study maximums of $\adpxrsym$, where the rotation is one bit left/right and one of the first two arguments is fixed. More precisely, we find $\beta', \gamma'$ such that $\max_{\beta, \gamma} \adpxr{\alpha}{\beta}{\gamma}{r} = \adpxr{\alpha}{\beta'}{\gamma'}{r}$, where $r = 1$ (one bit left rotation) and $r = n - 1$ (one bit right rotation).
One bit rotations are used, for instance, in CHAM~\cite{KooEtAl2017,RohEtAl2019} (see also COMET~\cite{GueronEtAl2019} which is based on CHAM and Speck) and Alzette~\cite{BeierleEtAl2020} (see also Sparkle~\cite{BeierleEtAl2020b}). 
Note that we do not analyze maximums for other argument fixations as well as for other rotations. These cases look more difficult. 
In addition, we obtain some symmetries of $\adpxrsym$.
They allow us to construct distinct differentials whose probabilities are the same.  
Finally, we provide all impossible differentials (i.e. differentials with probability zero) of the function $(x \oplus y) \lll r$ in terms of regular expression patterns. The lower and upper bounds for the numbers of them are $\frac{1}{8} 8^n$ and $\frac{5}{14} 8^n$ respectively. Similar number for the function $x \oplus y$ obtained in~\cite{LipmaaEtAl2004} is $\lfloor \frac{4}{7} 8^n \rfloor$. 
We can say that the rotation significantly reduces the number of impossible differentials.
Overall, the maximums $\max_{\beta, \gamma} \adpxr{\alpha}{\beta}{\gamma}{r}$ for $r = 1$ and $r = n - 1$ are very similar to maximums of $\adpsym$. More precisely, they coincide for the one bit left rotation (Corollary~\ref{cor:onebitleft_adp}) and coincide at least half the time for the one bit right rotation (Corollary~\ref{cor:onebitright_adp}). Also, the number of impossible differentials is maximal for the one bit left rotation. It may be an additional reason to consider these rotations as the simplest ones. 
``Maximal'' differentials are used to construct differential trails, whereas impossible differentials are of interest because of impossible differential cryptanalysis~\cite{Knudsen1998,BihamEtAl1999}. There are techniques for finding them for a wide class of transformations, see, for instance,~\cite{JongsungEtAl2003,LuoEtAl2014}. Some attacks (and approaches) against ARX ciphers can be found in~\cite{ChenEtAl2012,ZhangEtAl2018, AzimiEtAl2022}.

The paper is organized as follows. Necessary definitions are given in Section~\ref{sec:preliminaries}. 
In Section~\ref{sec:formulasection}, we introduce auxiliary notions ($\cadppar{c}$ and  $\padppar{a}{b}$) and rewrite the formula from~\cite{VelichkovEtAl2011} to obtain an expression for $\adpxrsym$ in these terms (Theorem~\ref{th:adrxrformula} and Corollary~\ref{cor:adrxrformula}).
Section~\ref{sec:symmetires} is devoted to the argument symmetries of $\adpxrsym$ which turn out to be similar to the symmetries of $\adpsym$ (Theorem~\ref{th:symmetries}).
Section~\ref{sec:maxfor1leftright} describes maximums of $\adpxrsym$, where the rotation is one bit left/right and one of the first two arguments is fixed (Theorem~\ref{th:max1left} in Section~\ref{sec:maxfor1} for the one bit left rotation and Theorem~\ref{th:max1right} in Section~\ref{sec:maxfor1right} for the one bit right rotation).
Section~\ref{sec:inconsistentdiffxr} provides all impossible differentials of the XR-operation (Theorem~\ref{th:adpXRzeros} and Table~\ref{table:zerosOfXR} in Section~\ref{sec:impossibledifferentialsdescr}) in terms of regular expression patterns (see Section~\ref{sec:patterns}). Some auxiliary results are moved to Appendix~\ref{sec:proofofpadpzeros}. Also, several estimations of the number $\nzer{n}{r}$ of tuples $(\alpha, \beta, \gamma) \in \zspc{n} \times \zspc{n} \times \zspc{n}$ such that $\adpxr{\alpha}{\beta}{\gamma}{r} = 0$ are given in Section~\ref{sec:impossibledifferentialsest} (Corollaries~\ref{cor:adpXRzeros_estimations} and \ref{cor:adpXRzeros_estimations_r_1}). The most significant of them are $\frac{1}{8} 8^n \leq \frac{1}{7} 8^n - \frac{1}{7} 8^r \leq \nzer{n}{r} < \lfloor \frac{5}{14} 8^n \rfloor = \nzer{n}{1}$, where $2 \leq r < n$.

\section{Preliminaries}\label{sec:preliminaries}
Let $x, y \in \zspc{n}$ be elements of the $n$-dimensional vector space over the two-element field,  
$x=(x_0,x_1, \ldots, x_{n-1})$.
Since ARX schemes mix modulo $2$ and modulo $2^n$ operations, we denote that  $x + y$, $x - y$ and $-x$ mean $x' + y' \mod{2^{n}}$, $x' - y' \mod{2^{n}}$ and $-x' \mod{2^{n}}$ respectively, where $$x' = x_0 2^{n - 1} + x_1 2^{n - 2} + ... + x_{n - 1} 2^{0}.$$ In other words, $x$ is the binary representation of the integer $x' \in \{0, \ldots, 2^n - 1\}$, where $x_0$ is its most significant bit. This is taken into account when we refer to~\cite{LipmaaEtAl2004} and \cite{VelichkovEtAl2011} since the most significant bit is $x_{n - 1}$ there.
The rotation is defined as $x \lll r = (x_{r}, \ldots, x_{n-1}, x_0, \ldots, x_{r - 1})$.
Let $\overline{x} = (x_0 \oplus 1, x_1 \oplus 1, \ldots, x_{n-1} \oplus 1).$
Recall that $\overline{x} = 2^n - 1 - x$.
For $a \in \zspc{}$ we define $\inv{x}{a}$ as $(x_0 \oplus a, x_1 \oplus a, \ldots, x_{n - 1} \oplus a)$, i.e. it is either $x$ or $\overline{x}$.
The Hamming weight (or shortly weight) $\wt{x}$ is the number of coordinates of $x$ that are not zero.
Also, $x \preceq y$ $\Longleftrightarrow$ $x_i \leq y_i$ for all $i \in \{0, \ldots, n - 1\}$. 

The concatenation $\cctn{\alpha}{\beta}$ is $(\alpha_0, \ldots, \alpha_{n-1}, \beta_0, \ldots, \beta_{m-1})$, where $\alpha \in \zspc{n}$ and $\beta \in \zspc{m}$.
If $a \in \zspc{}$, the concatenation $\cctn{\alpha}{a}$ is denoted by $\alpha a$. In terms of integers, $\alpha a = 2\alpha + a \pmod{2^{n+1}}$. 

We consider differences that are expressed using addition modulo $2^n$.
The additive differential probability $\adpgeneral{f}$ for $f: (\zspc{n})^k  \to \zspc{n}$ is defined as follows:
\begin{align*}   
    \adpgeneral{f}(\alpha_1, \ldots, \alpha_k \to \alpha_{k+1}) &= 2^{-kn} \#\{ x_1, \ldots, x_k  \in \zspc{n} \ : \\ &f(x_1 + \alpha_1, \ldots, x_k + \alpha_k) = f(x_1, \ldots, x_k) + \alpha_{k+1}\},
\end{align*}
where $\alpha_1, \ldots, \alpha_k \in \zspc{n}$ and $\alpha_{k+1} \in \zspc{n}$ are called input differences and an output difference respectively.

In this work, we consider $\adp{\alpha}{\beta}{\gamma}$ for the function $(x, y) \mapsto x \oplus y$ (see~\cite{LipmaaEtAl2004,MouhaEtAl2021}), $\adpxr{\alpha}{\beta}{\gamma}{r}$ for the XR-operation $(x, y) \mapsto (x \oplus y) \lll r$ and $\adprx{\alpha}{\beta}{\gamma}{r}$ for the RX-operation $(x, y) \mapsto (x \lll r) \oplus y$ (see~\cite{VelichkovEtAl2011}). 

Similarly to~\cite{LipmaaEtAl2004}, we use both binary vectors $\alpha, \beta, \gamma \in \zspc{n}$ and octal words $\omega(\alpha, \beta, \gamma) \in \{0, \ldots, 7\}^n$, where
$$
    (\omega(\alpha, \beta, \gamma))_i = 4\alpha_{i} + 2\beta_{i} + \gamma_{i} \text{ for all } i \in \{0, \ldots, n - 1\}.
$$
The weight of an octal symbol $\omega_i$ is the weight of the corresponding $(\alpha_i, \beta_i, \gamma_i)$.

Let $e_0, \ldots, e_7$ be the standard basis vectors of $\mathbb{Q}^8$, where $\mathbb{Q}$ is the field of rationals.
There is a rational series approach for calculating $\adpsym$. 
    \begin{theorem}[Lipmaa et al.~\cite{LipmaaEtAl2004}, 2004]\label{th:AdpMatrix}

 Let $L = (1, 1, 1, 1, 1, 1, 1, 1)$, $A_0, \ldots, A_7$ be $8\times8$ matrices, where
        \[
        A_0 = \frac{1}{4}\left(\begin{smallmatrix}
             4 & 0 & 0 & 1 & 0 & 1 & 1 & 0 \\
             0 & 0 & 0 & 1 & 0 & 1 & 0 & 0 \\
             0 & 0 & 0 & 1 & 0 & 0 & 1 & 0 \\
             0 & 0 & 0 & 1 & 0 & 0 & 0 & 0 \\
             0 & 0 & 0 & 0 & 0 & 1 & 1 & 0 \\
             0 & 0 & 0 & 0 & 0 & 1 & 0 & 0 \\
             0 & 0 & 0 & 0 & 0 & 0 & 1 & 0 \\
             0 & 0 & 0 & 0 & 0 & 0 & 0 & 0
        \end{smallmatrix}\right)
    \]
    and $(A_k)_{i,j} = (A_0)_{i \oplus k, j \oplus k}$, where $i, j, k \in \zspc{3}$. Then
    \[
        \adp{\alpha}{\beta}{\gamma} = \adpsym (\omega) = L A_{\omega_{0}} A_{\omega_{1}} \ldots A_{\omega_{n - 1}} e^T_0, 
    \]
    where $\alpha, \beta, \gamma \in \zspc{n}$ and $\omega = \omega(\alpha, \beta, \gamma)$.
    \end{theorem}

\begin{corollary}[Lipmaa et al.~\cite{LipmaaEtAl2004}, 2004]\label{cor:AdpMatrixZeros}
    Let $\alpha, \beta, \gamma \in \zspc{n}$, $\omega = \omega(\alpha, \beta, \gamma)$ and $k \in \{0, \ldots, n - 1\}$ such that $\omega_k \neq 0$ and $\omega_{k + 1} = \omega_{k + 2} = \ldots = \omega_{n - 1} = 0$. 
    If $\alpha = \beta = \gamma = 0$, we assume that $k = 0$.
    Then $\adp{\alpha}{\beta}{\gamma} = 0$ $\iff$ $\omega_k \in \{1, 2, 4, 7\}$. As a consequence,
    \begin{itemize}
        \item if $\omega_{n - 1} \in \{3, 5, 6\}$, then $\adp{\alpha}{\beta}{\gamma} > 0$;
        \item if $\omega_{n - 1} \in \{1, 2, 4, 7\}$, then $\adp{\alpha}{\beta}{\gamma} = 0$.
    \end{itemize}
    Note that $\omega_k \in \{0, 3, 5, 6\}$ $\iff$ $\alpha_k \oplus \beta_k \oplus \gamma_k = 0$ $\iff$ $\wt{\alpha_k, \beta_k, \gamma_k}$ is even, $\omega_k \in \{1, 2, 4, 7\}$ $\iff$ $\alpha_k \oplus \beta_k \oplus \gamma_k = 1$ $\iff$ $\wt{\alpha_k, \beta_k, \gamma_k}$ is odd.
\end{corollary}
    We use both integer and binary vector notations in the indexes of matrices, coordinates and octal words, i.e. the matrices $A_{(p_0,p_1,p_2)}$ and $A_{4p_0 + 2p_1 + p_2}$ (coordinates $v_{(p_0, p_1, p_2)}$ and $v_{4p_0 + 2p_1 + p_2}$ of $v \in \mathbb{Q}^8$) mean the same. Also, \cite[Proposition 1]{MouhaEtAl2021} provides that $\adpsym$ is symmetric, i.e. $\adp{\alpha}{\beta}{\gamma}$ = $\adp{\beta}{\alpha}{\gamma}$ = $\adp{\gamma}{\beta}{\alpha}$ and so on for any $\alpha, \beta, \gamma \in \zspc{n}$. We will refer to the results of~\cite{MouhaEtAl2021} taking into account this fact.

    \begin{example}\label{example:matriceslist}
    Theorem~\ref{th:AdpMatrix} allows us to calculate values of $\adpsym$ in the following way:
    \begin{align*}
        \adps{\oplus}{&\big(\stackrel{\alpha_0}{1},\stackrel{\alpha_1}{1},\stackrel{\alpha_2}{0},\stackrel{\alpha_3}{0}\big)}{\big(\stackrel{\beta_0}{0},\stackrel{\beta_1}{1},\stackrel{\beta_2}{1},\stackrel{\beta_3}{0}\big)}{\big(\stackrel{\gamma_0}{1},\stackrel{\gamma_1}{0},\stackrel{\gamma_2}{1},\stackrel{\gamma_3}{0}\big)} \\
        &= L \cdot A_{\big(\stackrel{\alpha_0}{1},\stackrel{\beta_0}{0},\stackrel{\gamma_0}{1}\big)} \cdot A_{\big(\stackrel{\alpha_1}{1},\stackrel{\beta_1}{1},\stackrel{\gamma_1}{0}\big)} \cdot A_{\big(\stackrel{\alpha_2}{0},\stackrel{\beta_2}{1},\stackrel{\gamma_2}{1}\big)} \cdot A_{\big(\stackrel{\alpha_3}{0},\stackrel{\beta_3}{0},\stackrel{\gamma_3}{0}\big)} \cdot e_0^T \\
        &= (1, 1, 1, 1, 1, 1, 1, 1) \cdot A_{5} \cdot A_{6} \cdot A_{3} \cdot A_{0} \cdot (1, 0, 0, 0, 0, 0, 0, 0)^T
        = \frac{16}{64} = \frac{1}{4}.
    \end{align*}
    Corollary~\ref{cor:AdpMatrixZeros} provides that this value is positive. Indeed, the first nonzero octal symbol (starting with the least significant bits) $(\alpha_2, \beta_2, \gamma_2) = (0, 1, 1)$ is of even weight.    
    Also, we list all matrices $A_0, \ldots, A_7$ in the explicit form.
    \begin{center}
    \begin{tabular}{cccc}
        $4 \cdot A_0$ & $4 \cdot A_1$ & $4 \cdot A_2$ & $4 \cdot A_3$ \\
        $\begin{psmallmatrix}
             4 & 0 & 0 & 1 & 0 & 1 & 1 & 0 \\
             0 & 0 & 0 & 1 & 0 & 1 & 0 & 0 \\
             0 & 0 & 0 & 1 & 0 & 0 & 1 & 0 \\
             0 & 0 & 0 & 1 & 0 & 0 & 0 & 0 \\
             0 & 0 & 0 & 0 & 0 & 1 & 1 & 0 \\
             0 & 0 & 0 & 0 & 0 & 1 & 0 & 0 \\
             0 & 0 & 0 & 0 & 0 & 0 & 1 & 0 \\
             0 & 0 & 0 & 0 & 0 & 0 & 0 & 0
        \end{psmallmatrix}$ &
       $\begin{psmallmatrix}
             0 & 0 & 1 & 0 & 1 & 0 & 0 & 0 \\
             0 & 4 & 1 & 0 & 1 & 0 & 0 & 1 \\
             0 & 0 & 1 & 0 & 0 & 0 & 0 & 0 \\
             0 & 0 & 1 & 0 & 0 & 0 & 0 & 1 \\
             0 & 0 & 0 & 0 & 1 & 0 & 0 & 0 \\
             0 & 0 & 0 & 0 & 1 & 0 & 0 & 1 \\
             0 & 0 & 0 & 0 & 0 & 0 & 0 & 0 \\
             0 & 0 & 0 & 0 & 0 & 0 & 0 & 1
        \end{psmallmatrix}$  &             
        $\begin{psmallmatrix}
             0 & 1 & 0 & 0 & 1 & 0 & 0 & 0 \\
             0 & 1 & 0 & 0 & 0 & 0 & 0 & 0 \\
             0 & 1 & 4 & 0 & 1 & 0 & 0 & 1 \\
             0 & 1 & 0 & 0 & 0 & 0 & 0 & 1 \\
             0 & 0 & 0 & 0 & 1 & 0 & 0 & 0 \\
             0 & 0 & 0 & 0 & 0 & 0 & 0 & 0 \\
             0 & 0 & 0 & 0 & 1 & 0 & 0 & 1 \\
             0 & 0 & 0 & 0 & 0 & 0 & 0 & 1
        \end{psmallmatrix}$ &               
       $\begin{psmallmatrix}
             1 & 0 & 0 & 0 & 0 & 0 & 0 & 0 \\
             1 & 0 & 0 & 0 & 0 & 1 & 0 & 0 \\
             1 & 0 & 0 & 0 & 0 & 0 & 1 & 0 \\
             1 & 0 & 0 & 4 & 0 & 1 & 1 & 0 \\
             0 & 0 & 0 & 0 & 0 & 0 & 0 & 0 \\
             0 & 0 & 0 & 0 & 0 & 1 & 0 & 0 \\
             0 & 0 & 0 & 0 & 0 & 0 & 1 & 0 \\
             0 & 0 & 0 & 0 & 0 & 1 & 1 & 0
        \end{psmallmatrix}$ \\
        $4 \cdot A_4$ & $4 \cdot A_5$ & $4 \cdot A_6$ & $4 \cdot A_7$ \\
        $\begin{psmallmatrix}
             0 & 1 & 1 & 0 & 0 & 0 & 0 & 0 \\
             0 & 1 & 0 & 0 & 0 & 0 & 0 & 0 \\
             0 & 0 & 1 & 0 & 0 & 0 & 0 & 0 \\
             0 & 0 & 0 & 0 & 0 & 0 & 0 & 0 \\
             0 & 1 & 1 & 0 & 4 & 0 & 0 & 1 \\
             0 & 1 & 0 & 0 & 0 & 0 & 0 & 1 \\
             0 & 0 & 1 & 0 & 0 & 0 & 0 & 1 \\
             0 & 0 & 0 & 0 & 0 & 0 & 0 & 1
        \end{psmallmatrix} $    &          
       $\begin{psmallmatrix}
             1 & 0 & 0 & 0 & 0 & 0 & 0 & 0 \\
             1 & 0 & 0 & 1 & 0 & 0 & 0 & 0 \\
             0 & 0 & 0 & 0 & 0 & 0 & 0 & 0 \\
             0 & 0 & 0 & 1 & 0 & 0 & 0 & 0 \\
             1 & 0 & 0 & 0 & 0 & 0 & 1 & 0 \\
             1 & 0 & 0 & 1 & 0 & 4 & 1 & 0 \\
             0 & 0 & 0 & 0 & 0 & 0 & 1 & 0 \\
             0 & 0 & 0 & 1 & 0 & 0 & 1 & 0
        \end{psmallmatrix} $ &
        $\begin{psmallmatrix}
             1 & 0 & 0 & 0 & 0 & 0 & 0 & 0 \\
             0 & 0 & 0 & 0 & 0 & 0 & 0 & 0 \\
             1 & 0 & 0 & 1 & 0 & 0 & 0 & 0 \\
             0 & 0 & 0 & 1 & 0 & 0 & 0 & 0 \\
             1 & 0 & 0 & 0 & 0 & 1 & 0 & 0 \\
             0 & 0 & 0 & 0 & 0 & 1 & 0 & 0 \\
             1 & 0 & 0 & 1 & 0 & 1 & 4 & 0 \\
             0 & 0 & 0 & 1 & 0 & 1 & 0 & 0
        \end{psmallmatrix}$  &
        $\begin{psmallmatrix}
             0 & 0 & 0 & 0 & 0 & 0 & 0 & 0 \\
             0 & 1 & 0 & 0 & 0 & 0 & 0 & 0 \\
             0 & 0 & 1 & 0 & 0 & 0 & 0 & 0 \\
             0 & 1 & 1 & 0 & 0 & 0 & 0 & 0 \\
             0 & 0 & 0 & 0 & 1 & 0 & 0 & 0 \\
             0 & 1 & 0 & 0 & 1 & 0 & 0 & 0 \\
             0 & 0 & 1 & 0 & 1 & 0 & 0 & 0 \\
             0 & 1 & 1 & 0 & 1 & 0 & 0 & 4
        \end{psmallmatrix}$
    \end{tabular}
    \end{center}
\end{example}

\section{A formula for \texorpdfstring{$\adpxrsym$}{}}\label{sec:formulasection}

In this section, we obtain a formula for $\adpxrsym$ by rewriting the formula from~\cite{VelichkovEtAl2011}. 
First of all, we introduce auxiliary 
$\cadppar{c}$ and $\padppar{a}{b}$ for $\alpha, \beta, \gamma \in \zspc{n}$ and $a, b, c \in \zspc{}$:
\begin{align*}
    \padp{a}{b}{\alpha}{\beta}{\gamma} &= L_{a,b} A_{\omega_{0}} \ldots A_{\omega_{n-1}} e^T_0,\\
    \cadp{c}{\alpha}{\beta}{\gamma} &= L_{c} A_{\omega_{0}} \ldots A_{\omega_{n-1}} e^T_0, 
    \end{align*}
    where $\omega = \omega(\alpha, \beta, \gamma)$,  $L_0 = (1, 0, 1, 0, 1, 0, 1, 0)$, $L_1 = (0, 1, 0, 1, 0, 1, 0, 1)$ and
    \begin{align*}
    L_{0,0} &= (1, 1, 0, 0, 0, 0, 0, 0), \ L_{0,1} = (0, 0, 1, 1, 0, 0, 0, 0),\\
    L_{1,0} &= (0, 0, 0, 0, 1, 1, 0, 0), \ L_{1,1} = (0, 0, 0, 0, 0, 0, 1, 1).
\end{align*}
Hence, $L_c = \sum_{a, b \in \zspc{}} e_{(a, b, c)}$ and $L_{a, b} = \sum_{c \in \zspc{}} e_{(a, b, c)}$. Theorem~\ref{th:AdpMatrix} implies that $$\sum_{a, b \in \zspc{}} \padp{a}{b}{\alpha}{\beta}{\gamma} = \sum_{c \in \zspc{}} \cadp{c}{\alpha}{\beta}{\gamma} = \adp{\alpha}{\beta}{\gamma}.$$

The values of $\cadppar{c}$ and $\padppar{a}{b}$ for $n = 1$ are given in Table~\ref{table:pcadp1bit}.
\begin{table}
\begin{center}
\begin{tabular}{@{}ccccccc@{}}
    \toprule
    $\alpha \beta \gamma$ & $\padppar{0}{0}$  & $\padppar{0}{1}$  & $\padppar{1}{0}$ & $\padppar{1}{1}$  & $\cadppar{0}$  & $\cadppar{1}$ \\  
    \midrule
    $000$ & $1$  & $0$  & $0$ & $0$ & $1$  & $0$ \\
    $011$ & $1/2$  & $1/2$  & $0$ & $0$ & $1/2$  & $1/2$ \\
    $101$ & $1/2$  & $0$  & $1/2$ & $0$ & $1/2$  & $1/2$ \\
    $110$ & $1/4$  & $1/4$  & $1/4$ & $1/4$ & $1$  & $0$ \\       
    \addlinespace
    $001$ & $0$  & $0$  & $0$ & $0$ & $0$  & $0$ \\
    $010$ & $0$  & $0$  & $0$ & $0$ & $0$  & $0$ \\
    $100$ & $0$  & $0$  & $0$ & $0$ & $0$  & $0$ \\
    $111$ & $0$  & $0$  & $0$ & $0$ & $0$  & $0$ \\       
    \bottomrule
\end{tabular}
\caption{$\padp{a}{b}{\alpha}{\beta}{\gamma}$ and $\cadp{c}{\alpha}{\beta}{\gamma}$, $\alpha, \beta, \gamma \in \zspc{}$}
\label{table:pcadp1bit}
\end{center}
\end{table}
Next, we prove the following lemma that is similar to~\cite[Lemma 4]{MouhaEtAl2021}.
\begin{lemma}\label{lemma:matrixindexXOR}
    Let $a, b, c \in \zspc{}$, and $k, w_0, \ldots , w_{n - 1} \in \zspc{3}$. Then
    \begin{itemize}
            \item $L_{c}A_{\omega_{0}} \ldots A_{\omega_{n - 1}} e^T_k = L_{c \oplus k_2} A_{\omega_{0} \oplus k} \ldots A_{\omega_{n - 1} \oplus k} e^T_0$,
            \item $L_{a,b}A_{\omega_{0}} \ldots A_{\omega_{n - 1}} e^T_k = L_{a \oplus k_0, b \oplus k_1} A_{\omega_{0} \oplus k} \ldots A_{\omega_{n - 1} \oplus k} e^T_0$.
    \end{itemize}
\end{lemma}
\begin{proof}
    Indeed, it is straightforward that $A_{t \oplus k} = T_k A_{t} T_k$, $e^T_k = T_k e^T_0$, $L_{c \oplus k_2} = L_c T_k$ and $L_{a \oplus k_0, b \oplus k_1} = L_{a,b} T_k$, where $t \in \zspc{3}$ and $T_k$ is the $8\times8$ involution matrix that swaps $i$ and $i \oplus k$ coordinates, $i \in \zspc{3}$.
\end{proof}

Now we rewrite the formula from~\cite{VelichkovEtAl2011} in terms of $\cadppar{c}$ and $\padppar{a}{b}$. 
\begin{theorem}\label{th:adrxrformula}
    Let $\alpha, \beta, \gamma \in \zspc{n - r}$, $\alpha', \beta', \gamma' \in \zspc{r}$. Then
    \begin{multline*}
    \adpxr{\cctn{\alpha'}{\alpha}}{\cctn{\beta'}{\beta}}{\cctn{\gamma}{\gamma'}}{r} \\ = \sum_{a,b,c \in \zspc{}} \padp{a}{b}{\alpha}{\beta}{\inv{\gamma}{c}} \cadp{c}{\inv{\alpha'}{a}}{\inv{\beta'}{b}}{\gamma'}.
    \end{multline*}
\end{theorem}
\begin{proof}
    Let $x = \cctn{\alpha'}{\alpha}$, $y = \cctn{\beta'}{\beta}$, $z = \cctn{\gamma}{\gamma'}$ and $r' = n - r$, i.e. $\alpha' = (x_0, \ldots, x_{r - 1})$, $\beta' = (y_0, \ldots, y_{r - 1})$ and $\gamma' = (z_{r'}, \ldots, z_{n - 1})$. According to~\cite[Theorem 1]{VelichkovEtAl2011}, 
    $$
        \adpr{\mathrm{RX}}{z}{y}{x}{r'} = \frac{1}{4^{n}}\sum\limits_{j \in \{0, 2, 4, 6\}} L'_j A'_{w[n-1]} \ldots A'_{w[n - r']} R A'_{w[n - r' - 1]} \ldots A'_{w[0]} C_j,
    $$  
    where the denotations are given in~\cite[\textsection 5.4]{VelichkovEtAl2011} in the following form:
    \begin{itemize}
        \item $w[n - i - 1] = z_i y_{i - r'} x_{i - r'} = z_i y_{i + r} x_{i + r}$, i.e. it is one of $000, 001, \ldots, 111$, $i \in \{0, \ldots, n-1\}$, where $\mathrm{mod}\ n$ is used for all indices; since the least significant coordinate of $x, y$ and $z$ is $n - 1$, we have changed the original $w[i] = z_i y_{i + r'} x_{i + r'}$;
        \item $L'_0 = L_{0,0}$, $L'_2 = L_{0,1}$, $L'_4 = L_{1,0}$, $L'_6 = L_{1,1}$ and $C_j = e^T_j$;
        \item $A'_{abc}$ and $R$, where $a, b, c \in \zspc{}$, are $8 \times 8$ matrices given in \cite[Appendix A]{VelichkovEtAl2011} ($A'_{abc}$ are denoted by $A_{abc}$ there). It can be seen that $$A'_{a b c} = 4 A_{(\overline{c}, b, a)} = 4A_{4\overline{c} + 2b + a} \text{ and } R = e^T_4 L_0 + e^T_5 L_1.$$
    \end{itemize}    
    Note that $\adprxsym$ is denoted by $\adpgeneral{\mathrm{ARX}}$ in~\cite{VelichkovEtAl2011} since they can be transformed to each other by changing the first argument.

    Let $v_i = (\overline{x_{i + r}}, y_{i + r}, z_i) \in \zspc{3}$ for all $i \in \{0, \ldots, n-1\}$, the indices are $\mathrm{mod}\ n$ as well. Thus, we can replace $A'_{w[n - i - 1]}$ by $4A_{v_{i}}$ and  move to our denotations:
\begin{align*}
        &\adpxr{\cctn{\alpha'}{\alpha}}{\cctn{\beta'}{\beta}}{\cctn{\gamma}{\gamma'}}{r} = \adpxr{x}{y}{z}{r} = \adpr{\mathrm{RX}}{z}{y}{x}{r'}\\
        &= \sum\limits_{j \in \{0, 2, 4, 6\}} L'_j A_{v_{0}} \ldots A_{v_{r' - 1}} (e^T_4 L_0 + e^T_5 L_1) A_{v_{r'}} \ldots A_{v_{n - 1}} e^T_j\\
        &= \sum\limits_{j \in \{0, 2, 4, 6\}} \big(L'_j A_{v_{0}} \ldots A_{v_{r' - 1}} e^T_4 \cdot L_0 A_{v_{r'}} \ldots A_{v_{n - 1}} e^T_j \\
        &\omit\hfill\ensuremath{ 
        + L'_j A_{v_{0}} \ldots A_{v_{r' - 1}} e^T_5 \cdot L_1 A_{v_{r'}} \ldots A_{v_{n - 1}} e^T_j\big)}\\
        &= \sum_{a,b,c \in \zspc{}} L_{a,b} A_{v_{0}} \ldots A_{v_{r' - 1}} e^T_{(1, 0, c)} \cdot L_c A_{v_{r'}} \ldots A_{v_{n - 1}} e^T_{(a, b, 0)}\\
        &\stackrel{*}{=} \sum_{a,b,c \in \zspc{}} L_{\overline{a},b} A_{v_{0} \oplus (1, 0, c)} \ldots A_{v_{r' - 1} \oplus (1, 0, c)} e^T_{0} \cdot L_c A_{v_{r'} \oplus (a, b, 0)} \ldots A_{v_{n - 1} \oplus (a, b, 0)} e^T_0\\       
        &\stackrel{}{=} \sum_{a,b,c \in \zspc{}} L_{a,b} A_{v_{0} \oplus (1, 0, c)} \ldots A_{v_{r' - 1} \oplus (1, 0, c)} e^T_{0} \cdot L_c A_{v_{r'} \oplus (\overline{a}, b, 0)} \ldots A_{v_{n - 1} \oplus (\overline{a}, b, 0)} e^T_0,
   \end{align*}
   where Lemma~\ref{lemma:matrixindexXOR} is applied twice at the step $*$ and $a$ is replaced by $\overline{a}$ at the last step. 
   Finally, we can see that 
    \begin{align*}
        \cadp{c}{\inv{\alpha'}{a}}{\inv{\beta'}{b}}{\gamma'} &= L_c A_{v_{r'} \oplus (a \oplus 1, b, 0)} \ldots A_{v_{n - 1} \oplus (a \oplus 1, b, 0)} e^T_0 \text{ and } \\
        \padp{a}{b}{\alpha}{\beta}{\inv{\gamma}{c}} &= L_{a,b} A_{v_{0} \oplus (1, 0, c)} \ldots A_{v_{r' - 1} \oplus (1, 0, c)} e^T_{0},
    \end{align*}        
    since $v_{0}, \ldots, v_{r' - 1}$ and $v_{r'}, \ldots, v_{n - 1}$ are equal to
    $$
        (\overline{x_{r}}, y_r, z_0), \ldots, (\overline{x_{n - 1}}, y_{n - 1}, z_{r' - 1}) \text{ and } (\overline{x_{0}}, y_{0}, z_{r'}), \ldots, (\overline{x_{r - 1}}, y_{r - 1}, z_{n - 1}).
    $$
    The theorem is proved.
\end{proof}

\begin{corollary}\label{cor:adrxrformula}
    Let $\alpha, \beta, \gamma \in \zspc{n - r}$, $\alpha', \beta', \gamma' \in \zspc{r}$, $a = \alpha_{n - r - 1} \oplus \beta_{n - r - 1} \oplus \gamma_{n - r - 1}$ and $a' = \alpha'_{r - 1} \oplus \beta'_{r - 1} \oplus \gamma'_{r - 1}$. Then
\begin{multline*}
     \adpxr{\cctn{\alpha'}{\alpha}}{\cctn{\beta'}{\beta}}{\cctn{\gamma}{\gamma'}}{r} = \padp{a'}{0}{\alpha}{\beta}{\inv{\gamma}{a}} \cadp{a}{\inv{\alpha'}{a'}}{\beta'}{\gamma'}\\
     + \padp{\overline{a'}}{1}{\alpha}{\beta}{\inv{\gamma}{a}} \cadp{a}{\inv{\overline{\alpha'}}{a'}}{\overline{\beta'}}{\gamma'}.
\end{multline*}
\end{corollary}
\begin{proof}
    By Theorem~\ref{th:adrxrformula}, $\adpxr{\cctn{\alpha'}{\alpha}}{\cctn{\beta'}{\beta}}{\cctn{\gamma}{\gamma'}}{r}$ is equal to
    $$
        \sum_{p,q,u \in \zspc{}} \padp{p}{q}{\alpha}{\beta}{\inv{\gamma}{u}} \cadp{u}{\inv{\alpha'}{p}}{\inv{\beta'}{q}}{\gamma'}.
    $$
    If $\alpha_{n - r - 1} \oplus \beta_{n - r - 1} \oplus \gamma_{n - r - 1} = 1$, then $\adp{\alpha}{\beta}{\gamma} = 0$ by Corollary~\ref{cor:AdpMatrixZeros}. It means that $\padp{p}{q}{\alpha}{\beta}{\gamma} = \cadp{u}{\alpha}{\beta}{\gamma} = 0$.
    Hence, we can exclude $u \neq \alpha_{n - r - 1} \oplus \beta_{n - r - 1} \oplus \gamma_{n - r - 1} = a$. Also, we can consider only
    \begin{align*}
        \alpha'_{r - 1} \oplus p \oplus \beta'_{r - 1} \oplus q \oplus \gamma'_{r - 1} = 0 &\Rightarrow
        \begin{cases}
            p = \alpha'_{r - 1} \oplus \beta'_{r - 1} \oplus \gamma'_{r - 1} = a', & \text{if } q = 0, \\
            p = \overline{\alpha'_{r - 1} \oplus \beta'_{r - 1} \oplus \gamma'_{r - 1}} = \overline{a'} , & \text{if } q = 1.
        \end{cases}
    \end{align*}
    This completes the proof.
\end{proof}

Auxiliary $\cadppar{c}$ and $\padppar{a}{b}$ can be also calculated using recurrence formulas that are similar to ones for $\adpsym$ proposed in~\cite[Theorem 3]{MouhaEtAl2021}.
  \begin{theorem}\label{th:pcadprecurrence}
Let $\alpha, \beta, \gamma \in \zspc{n}, p \in \zspc{3}, a, b, c \in \zspc{}$ and $\wt{p}$ be even. Then
\begin{align*}
        \cadp{c}{\alpha p_0}{\beta p_1}{\gamma p_2} &=
            \frac{1}{2^{\wt{p}}}\sum\limits_{q \in \zspc{3}, q \preceq p}\cadp{c \oplus q_2}{\inv{\alpha}{q_0}}{\inv{\beta}{q_1}}{\inv{\gamma}{q_2}},\\
        \padp{a}{b}{\alpha p_0}{\beta p_1}{\gamma p_2} &=
            \frac{1}{2^{\wt{p}}}\sum\limits_{q \in \zspc{3}, q \preceq p}\padp{a \oplus q_0}{b \oplus q_1}{\inv{\alpha}{q_0}}{\inv{\beta}{q_1}}{\inv{\gamma}{q_2}}.
    \end{align*}
    If $\wt{p}$ is odd, $\cadp{c}{\alpha p_0}{\beta p_1}{\gamma p_2} = \padp{a}{b}{\alpha p_0}{\beta p_1}{\gamma p_2} = 0$.
\end{theorem}
\begin{proof}
    By definition, 
    $
        \cadp{c}{\alpha p_0}{\beta p_1}{\gamma p_2} = L_{c} A_{\omega_{0}} \ldots A_{\omega_{n-1}} A_{p} e^T_0.
    $
    First of all, $\adp{\alpha p_0}{\beta p_1}{\gamma p_2} = 0$ if $\wt{p}$ is odd (see Corollary~\ref{cor:AdpMatrixZeros}). This implies the same for $\cadppar{c}$ and $\padppar{a}{b}$. Next, let $\wt{p}$ be even. It can be seen that 
    $$
        A_{p} e^T_0 = \frac{1}{2^{\wt{p}}}\sum\limits_{q \in \zspc{3}, q \preceq p} e^T_q.
    $$
    Therefore,
    \begin{align*}
        \cadp{c}{&\alpha p_0}{\beta p_1}{\gamma p_2} = \frac{1}{2^{\wt{p}}}\sum\limits_{q \in \zspc{3}, q \preceq p} L_{c} A_{\omega_{0}} \ldots A_{\omega_{n-1}} e^T_q\\
        &{=} \frac{1}{2^{\wt{p}}}\sum\limits_{q \in \zspc{3}, q \preceq p} L_{c \oplus q_2} A_{\omega_{0} \oplus q} \ldots A_{\omega_{n-1} \oplus q} e^T_0 \ \ \ \ \ \ \  (\mathrm{Lemma}~\ref{lemma:matrixindexXOR})\\
        &{=} \frac{1}{2^{\wt{p}}}\sum\limits_{q \in \zspc{3}, q \preceq p} L_{c \oplus q_2} A_{(\alpha_{0}, \beta_{0}, \gamma_{0})  \oplus (q_0, q_1, q_2)} \ldots A_{(\alpha_{n-1}, \beta_{n-1}, \gamma_{n-1}) \oplus (q_0, q_1, q_2)} e^T_0\\
        &{=} \frac{1}{2^{\wt{p}}}\sum\limits_{q \in \zspc{3}, q \preceq p}\cadp{c \oplus q_2}{\inv{\alpha}{q_0}}{\inv{\beta}{q_1}}{\inv{\gamma}{q_2}}.
    \end{align*}
    The recurrence formula for $\padppar{a}{b}$ can be proved in the same way.
\end{proof}  

Thus, Table~\ref{table:pcadp1bit} and Theorem~\ref{th:pcadprecurrence} can be the induction base and the induction step respectively for calculating both $\cadppar{c}$ and $\padppar{a}{b}$. As a consequence, we can use them for calculating $\adpxrsym$.

\section{Symmetries of \texorpdfstring{$\adpxrsym$}{}}\label{sec:symmetires}
In this section, we prove some argument symmetries of $\adpxrsym$. This can be helpful to construct optimal differences, impossible differences and so on since they may provide distinct differences with the same probability.

First, we prove argument symmetries of $\cadppar{c}$ and $\padppar{a}{b}$.

\begin{lemma}\label{lemma:forMinuses}
    Let $\omega = \omega(\alpha, \beta, \gamma)$ and $\omega' = \omega(-\alpha, \beta, \gamma)$, where $\alpha, \beta, \gamma \in \zspc{n}$ and $\alpha \neq 0$. Then for all $i \in \{0, \ldots, 7\}$ the following holds:
    $$
        e_i A_{\omega_{0}} \ldots A_{\omega_{n - 1}} e_0^T = e_{i \oplus 4} A_{\omega'_{0}} \ldots A_{\omega'_{n - 1}} e_0^T.
    $$
\end{lemma}
\begin{proof}
    It is clear that $\alpha =  (\alpha_0, \ldots, \alpha_{m - 1}, 1, 0, \ldots, 0)$ for some $m$, $0 \leq m \leq n - 1$, where $m = 0$ means that $\alpha = 2^{n - 1}$. Hence,
    $$
        -\alpha = (\overline{\alpha_{0}}, \overline{\alpha_{1}}, \ldots, \overline{\alpha_{m-1}}, 1, 0, \ldots, 0).
    $$
    Also, $\omega_{i} = \omega'_{i}$ for all $i \in \{m, \ldots, n - 1\}$.
    
    Let $T_4$ be the involution matrix that swaps $j$ and $j \oplus 4$ coordinates, $j \in \zspc{3}$. Then
    \begin{align*}
        e_{i \oplus 4} A_{\omega'_{0}} \ldots A_{\omega'_{n - 1}} e_0^T &= e_{i \oplus 4} (T_4 T_4) A_{\omega'_{0}} (T_4 T_4) \ldots (T_4 T_4) A_{\omega'_{m - 1}} (T_4 T_4) v\\ 
        &= (e_{i \oplus 4} T_4) (T_4 A_{\omega'_{0}} T_4) \ldots (T_4 A_{\omega'_{m - 1}} T_4) (T_4 v) \\ 
        &= e_i A_{\omega_{0}} \ldots A_{\omega_{m - 1}} (T_4 v), 
    \end{align*}
    where $v = A_{\omega_{m}} \ldots A_{\omega_{n - 1}} e_0^T$. Thus, we only need to show that $T_4 v = v$ that is equivalent to $v_s = v_{s \oplus 4}$ for any $s \in \{0, \ldots, 7\}$.
 
    Let $u = A_{\omega_{m + 1}} \ldots A_{\omega_{n - 1}} e_0^T$, where  $u = e_0^T$ if $m = n - 1$. First of all, $\omega_{m + 1}, \ldots, \omega_{n - 1} \in \{0, 1, 2, 3\}$. It is not difficult to see that any element of $A_0, \ldots, A_3$ whose indices $i \in \{4, 5, 6, 7\}$ and $j \in \{0, 1, 2, 3\}$ is zero. The coordinates $4, 5, 6$ and $7$ of $e_0$ are also zeros. This means that $u_4 = u_5 = u_6 = u_7 = 0$.

    Next, we denote $\omega_{m}$ by $k$, i.e. $v = A_k u$. It is clear that $k \in \{4, 5, 6, 7\}$. According to the definition of the matrix $A_k$ (see Theorem~\ref{th:AdpMatrix}),
    \begin{align*}
        &v_{k} = u_{k} + \frac{1}{4} u_{k \oplus 3} + \frac{1}{4} u_{k \oplus 5} + \frac{1}{4} u_{k \oplus 6}, &&v_{k \oplus 4} = \frac{1}{4} u_{k \oplus 5} + \frac{1}{4} u_{k \oplus 6}, \\
        &v_{k \oplus 1} = \frac{1}{4} u_{k \oplus 3} + \frac{1}{4} u_{k \oplus 5},
        &&v_{k \oplus 5} = \frac{1}{4} u_{k \oplus 5},        
        \\
        &v_{k \oplus 2} = \frac{1}{4} u_{k \oplus 3} + \frac{1}{4} u_{k \oplus 6},
        &&v_{k \oplus 6} = \frac{1}{4} u_{k \oplus 6},
        \\
        &v_{k \oplus 3} = \frac{1}{4} u_{k \oplus 3},
        &&v_{k \oplus 7} = 0.
    \end{align*}
    Both $k$ and $k \oplus 3$ belong to $\{4, 5, 6, 7\}$ since $k \geq 4$. It means that $u_k = u_{k \oplus 3} = 0$. Thus,
    \begin{align*}
        v_{k} = v_{k \oplus 4} &= \frac{1}{4} u_{k \oplus 5} + \frac{1}{4} u_{k \oplus 6},\\
        v_{k \oplus 1} = v_{k \oplus 5} &= \frac{1}{4} u_{k \oplus 5},\\
        v_{k \oplus 2} = v_{k \oplus 6} &= \frac{1}{4} u_{k \oplus 6},\\
        v_{k \oplus 3} = v_{k \oplus 7} &= 0.
    \end{align*} 
    Since $\{ k, k \oplus 1, k \oplus 2, k \oplus 3, k \oplus 4, k \oplus 5, k \oplus 6, k \oplus 7 \} = \{0, \ldots, 7\}$, we obtain that $v_s = v_{s \oplus 4}$ for any $s \in \{0, \ldots, 7\}$. The lemma is proved.
\end{proof}
The lemma allows us to establish some argument symmetries of $\cadppar{c}$ and $\padppar{a}{b}$.
\begin{theorem}\label{th:cadpsymmetries}\label{th:padpsymmetries}
Let $\alpha, \beta, \gamma \in \zspc{n}$, $a, b, c \in \zspc{}$. Then
\begin{enumerate}
    \item $\cadp{c}{\alpha}{\beta}{\gamma} = \cadp{c}{\beta}{\alpha}{\gamma} =  \cadp{c}{\alpha + 2^{n-1}}{\beta + 2^{n-1}}{\gamma}$,
    \item $\cadp{c}{\alpha}{\beta}{\gamma} = \cadp{c}{-\alpha}{\beta}{\gamma} = \cadp{c}{\alpha}{-\beta}{\gamma}$, 
    \item $\cadp{c}{\alpha}{\beta}{\gamma} = \cadp{c \oplus 1}{\alpha}{\beta}{-\gamma}$ for $\gamma \neq 0$,
    \item $\padp{a}{b}{\alpha}{\beta}{\gamma} = \padp{a}{b}{\alpha}{\beta}{-\gamma}$, 
    \item $\padp{a}{b}{-\alpha}{\beta}{\gamma} = \padp{a \oplus 1}{b}{\alpha}{\beta}{\gamma}$ for $\alpha \neq 0$,
    \item $\padp{a}{b}{\alpha}{-\beta}{\gamma} = \padp{a}{b \oplus 1}{\alpha}{\beta}{\gamma}$ for $\beta \neq 0$,
    \item $\padp{1}{b}{0}{\beta}{\gamma} = \padp{a}{1}{\alpha}{0}{\gamma} = \cadp{1}{\alpha}{\beta}{0} = 0$.
\end{enumerate}    
\end{theorem}
\begin{proof}
    Let $\omega = \omega(\alpha, \beta, \gamma)$, the mapping $B$ swap bits of $\alpha$ and $\beta$, the mapping $\Gamma$ swap bits of $\alpha$ and $\gamma$, i.e. $B(a, b, c) = (b, a, c)$ and $\Gamma(a, b, c) = (c, b, a)$ for any $a, b, c \in \zspc{}$.     
    Let us use the following notation:
    $$
        v = A_{\omega_{0}} \ldots A_{\omega_{n-1}} e_0^T, \ v_{B} = A_{B(\omega_{0})} \ldots A_{B(\omega_{n-1})} e_0^T, \  v_{\Gamma} = A_{\Gamma(\omega_{0})} \ldots A_{\Gamma(\omega_{n-1})} e_0^T.
    $$
    Also, it is not difficult to see that
    $$
        L_c = e_{c} + e_{c \oplus 2} + e_{c \oplus 4} + e_{c \oplus 6} \text{ and } L_{a,b} = e_{4a \oplus b} + e_{4a \oplus b \oplus 1}.
    $$   
    
    1. By definition,
    $$
        \cadp{c}{\beta}{\alpha}{\gamma} = L_c v_{B} \text{ and } \cadp{c}{\alpha}{\beta}{\gamma} = L_c v.
    $$
    According to Lemma~\ref{lemma:psymmetries} (see Appendix~\ref{sec:proofofpadpzeros}), 
    \begin{align*}
       L_c v &=  (e_{B(c)} + e_{B(c \oplus 2)} + e_{B(c \oplus 4)} + e_{B(c \oplus 6)}) v_{B} \\
       &=  (e_{c} + e_{c \oplus 4} + e_{c \oplus 2} + e_{c \oplus 6}) v_{B} = L_c v_{B},
    \end{align*}
    i.e. $\cadp{c}{\beta}{\alpha}{\gamma} = \cadp{c}{\alpha}{\beta}{\gamma}$. Next, by definition,
    \begin{align*}
        \cadp{c}{\alpha + 2^{n-1}}{\beta + 2^{n - 1}}{\gamma} &= L_c A_{\omega_{0} \oplus 6} A_{\omega_{1}} \ldots A_{\omega_{n - 1}} e_0^T. 
    \end{align*}
    We only need to show that $L_c A_{t \oplus 6} = L_c A_{t}$ for any $t \in \zspc{3}$. 
    Indeed,
    $$
        L_c A_{t \oplus 6} = (((L_c T_{t \oplus 6}) A_0) T_6) T_{t} = ((L_{c \oplus t_2} A_0) T_6) T_{t},
    $$
    where $T_k$ is the involution matrix that swaps $i$ and $i \oplus k$ coordinates, $i, k \in \zspc{3}$.
    At the same time, $L_c A_{t} = (L_{c \oplus t_2} A_0)T_{t}$. 
    Finally,    
    $$
        L_0 A_0 = (1, 0, 0, \frac{1}{2}, 0, \frac{1}{2}, 1, 0) \text{ and }
        L_1 A_0 = (0, 0, 0, \frac{1}{2}, 0, \frac{1}{2}, 0, 0),
    $$
    that implies $(L_{c \oplus t_2} A_0) T_6 = L_{c \oplus t_2} A_0$.

    2-6. Let $\omega^{\alpha} = \omega(-\alpha, \beta, \gamma)$, $\omega^{\beta} = \omega(\alpha, -\beta, \gamma)$ and $\omega^{\gamma} = \omega(\alpha, \beta, -\gamma)$. Hereinafter, we use $\omega^{\alpha}$ if $\alpha \neq 0$, $\omega^{\beta}$ if $\beta \neq 0$ and $\omega^{\gamma}$ if $\gamma \neq 0$. Otherwise, there is nothing to prove since either $\omega^{\alpha} = \omega$ or $\omega^{\beta} = \omega$ or $\omega^{\gamma} = \omega$).     
    Giving $\omega' = \omega(-\beta, \alpha, \gamma)$, we obtain that
    \begin{align*}
        e_i A_{\omega_{0}} \ldots A_{\omega_{n - 1}} e_0^T &= e_{B(i)} A_{B(\omega_{0})} \ldots A_{B(\omega_{n - 1})} e_0^T &&\text{ (Lemma~\ref{lemma:psymmetries})} \\
        &= e_{B(i) \oplus 4} A_{\omega'_{0}} \ldots A_{\omega'_{n - 1}} e_0^T &&\text{ (Lemma~\ref{lemma:forMinuses})}\\
        &= e_{B(B(i) \oplus 4)} A_{B(\omega'_{0})} \ldots A_{B(\omega'_{n - 1})} e_0^T &&\text{ (Lemma~\ref{lemma:psymmetries})}\\
        &= e_{i \oplus 2} A_{\omega^{\beta}_{0}} \ldots A_{\omega^{\beta}_{n - 1}} e_0^T.
    \end{align*}
    We can show in the same way that $e_i A_{\omega_{0}} \ldots A_{\omega_{n - 1}} e_0^T = e_{i \oplus 1} A_{\omega^{\gamma}_{0}} \ldots A_{\omega^{\gamma}_{n - 1}} e_0^T$.
    Finally,
    \begin{align*}
        e_{c \oplus 4} + e_{(c \oplus 2) \oplus 4} + e_{(c \oplus 4) \oplus 4} + e_{(c \oplus 6) \oplus 4} &= L_c,\\
        e_{c \oplus 2} + e_{(c \oplus 2) \oplus 2} + e_{(c \oplus 4) \oplus 2} + e_{(c \oplus 6) \oplus 2} &= L_c,\\
        e_{c \oplus 1} + e_{(c \oplus 2) \oplus 1} + e_{(c \oplus 4) \oplus 1} + e_{(c \oplus 6) \oplus 1} &= L_{c \oplus 1},
    \end{align*}    
    i.e. $\cadp{c}{\alpha}{\beta}{\gamma} = \cadp{c}{-\alpha}{\beta}{\gamma} = \cadp{c}{\alpha}{-\beta}{\gamma}$ and $\cadp{c}{\alpha}{\beta}{\gamma} = \cadp{c \oplus 1}{\alpha}{\beta}{-\gamma}$ for $\gamma \neq 0$. 
    Similarly,
    \begin{align*}
        e_{(4a \oplus 2b) \oplus 4} + e_{(4a \oplus 2b \oplus 1) \oplus 4}  &= L_{a \oplus 1, b},\\
        e_{(4a \oplus 2b) \oplus 2} + e_{(4a \oplus 2b \oplus 1) \oplus 2}  &= L_{a, b \oplus 1},\\
        e_{(4a \oplus 2b) \oplus 1} + e_{(4a \oplus 2b \oplus 1) \oplus 1}  &= L_{a, b},
    \end{align*}
    i.e. $\padp{a}{b}{\alpha}{\beta}{\gamma} = \padp{a}{b}{\alpha}{\beta}{-\gamma}$, $\padp{a}{b}{\alpha}{\beta}{\gamma} = \padp{a \oplus 1}{b}{-\alpha}{\beta}{\gamma}$ for $\alpha \neq 0$ and $\padp{a}{b}{\alpha}{\beta}{\gamma} = \padp{a}{b \oplus 1}{\alpha}{-\beta}{\gamma}$ for $\beta \neq 0$. 
    
    7. Let $\alpha = 0$. As it was mentioned in the proof of Lemma~\ref{lemma:forMinuses}, $\omega_{0}, \ldots, \omega_{n - 1} \in \{0, 1, 2, 3\}$ and any element of $A_0, \ldots, A_3$ whose indices $i \in \{4, 5, 6, 7\}$ and $j \in \{0, 1, 2, 3\}$ is zero, the coordinates $4, 5, 6$ and $7$ of $e_0$ are zero as well. This means that $v_4 = v_5 = v_6 = v_7 = 0$. Similar to the previous points, mappings $B$ and $\Gamma$ guarantee that $v_2 = v_3 = v_6 = v_7 = 0$ if $\beta = 0$ and $v_1 = v_3 = v_5 = v_7 = 0$ if $\gamma = 0$. Thus, $L_{1, b} v = 0$, $L_{a, 1} v = 0$ and $L_1 v = 0$ if $\alpha = 0$, $\beta = 0$ and $\gamma = 0$ respectively.
\end{proof}

Similar properties of $\adpxrsym$ easily follows from the proved statements.

\begin{theorem}\label{th:symmetries}
The values of $\mathrm{adp}^{\mathrm{XR}}$ satisfy the following properties:
\begin{itemize}
    \item $\adpxr{\alpha}{\beta}{\gamma}{r} = \adpxr{\beta}{\alpha}{\gamma}{r}$,
    \item $\adpxr{\alpha}{\beta}{\gamma}{r} = \adpxr{\alpha + 2^{n-1}}{\beta + 2^{n-1}}{\gamma}{r}
    $,
    \item $\adpxr{\alpha}{\beta}{\gamma}{r} = \adpxr{\pm\alpha}{\pm\beta}{\pm\gamma}{r}$, where $\pm \alpha$ means that we can substitute either $\alpha$ or $-\alpha$.
\end{itemize}
\end{theorem}
\begin{proof}
    The first two points are straightforward by definition (note that $\alpha + 2^{n - 1} = \alpha \oplus 2^{n - 1}$). Next, 
    we do not consider $\beta$ since $\alpha$ and $\beta$ are symmetric.
    Let $\alpha, \beta, \gamma \in \zspc{n - r}$, $\alpha', \beta', \gamma' \in \zspc{r}$. According to Theorem~\ref{th:adrxrformula}, $\adpxr{\cctn{\alpha'}{\alpha}}{\cctn{\beta'}{\beta}}{\cctn{\gamma}{\gamma'}}{r}$ is equal to
    \begin{equation}\label{eq:fromthxr}
        \sum_{a,b,c \in \zspc{}} \padp{a}{b}{\alpha}{\beta}{\inv{\gamma}{c}} \cadp{c}{\inv{\alpha'}{a}}{\inv{\beta'}{b}}{\gamma'}.
    \end{equation}
    
    \textbf{Case 1.} $\alpha \neq 0$ (resp. $\gamma' \neq 0$). 
    Thus, $-\cctn{\alpha'}{\alpha} = \cctn{\overline{\alpha'}}{-\alpha}$. Let us substitute it to~(\ref{eq:fromthxr}):  $\cadp{c}{\inv{\overline{\alpha'}}{a}}{\inv{\beta'}{b}}{\gamma'} = \cadp{c}{\inv{\alpha'}{a \oplus 1}}{\inv{\beta'}{b}}{\gamma'}$. Also, $\padp{a}{b}{-\alpha}{\beta}{\inv{\gamma}{c}} = \padp{a \oplus 1}{b}{\alpha}{\beta}{\inv{\gamma}{c}}$ by Theorem~\ref{th:cadpsymmetries}. But $(\ref{eq:fromthxr})$ does not change if we replace $a$ by $a \oplus 1$. Similar reasons provide the same result for $-\cctn{\gamma}{\gamma'} = \cctn{\overline{\gamma}}{-\gamma'}$.
        
    \textbf{Case 2.} $\alpha = 0$ (resp. $\gamma' = 0$). 
    Thus, $-\cctn{\alpha'}{0} = \cctn{-\alpha'}{0}$. Substituting it to~(\ref{eq:fromthxr}), $\padp{1}{b}{0}{\beta}{\inv{\gamma}{c}} = 0$ (Theorem~\ref{th:padpsymmetries}). Hence, we can only consider $a = 0$. But in this case $\cadppar{c}$ does not contain $\overline{\alpha'}$. By Theorem~\ref{th:cadpsymmetries},  
    $\cadp{c}{-\alpha'}{\inv{\beta'}{b}}{\gamma'} = \cadp{c}{\alpha'}{\inv{\beta'}{b}}{\gamma'}$. This means that  $(\ref{eq:fromthxr})$ does not change.
    Similar reasons provide the same result for $-\cctn{\gamma}{0} = \cctn{-\gamma}{0}$.
\end{proof}

These symmetries belong to the argument symmetries of $\adpsym$. At the same time, we can swap any two arguments of $\adpsym$ which does not work for $\adpxrsym$. 

\section{Maximums of \texorpdfstring{$\adpxrsym$}{} for one bit rotations}\label{sec:maxfor1leftright}

In this section, we consider maximums of $\adpxrsym$ for $r=1$ (the rotation is one bit left) and for $r=n-1$ (the rotation is one bit right), where one of the first two arguments is fixed. Due to the symmetries provided by Theorem~\ref{th:symmetries}, nothing is changed if we fix the second argument of $\adpxrsym$ instead of the first one.
Such maximums may be helpful for constructing optimal differential trails or estimating their probabilities. 

\subsection{Maximums of \texorpdfstring{$\adpxrsym$}{} for \texorpdfstring{$r=1$}{}}\label{sec:maxfor1}
Using the recurrence formulas
, we prove the following lemma.
\begin{lemma}\label{lemma:padpsum}
    Let $\alpha, \beta, \gamma \in \zspc{n}$, $a \in \zspc{}$. Then the following holds:
    \begin{multline*}
        \padp{a}{0}{\alpha}{\beta}{\gamma} + \padp{\overline{a}}{1}{\alpha}{\beta}{\gamma} 
        \leq \padp{0}{0}{\alpha}{\alpha}{0} + \padp{1}{1}{\alpha}{\alpha}{0}.
    \end{multline*}
\end{lemma}
\begin{proof}
    Let us prove the statement by induction. The base of the induction $n = 1$ directly follows from Table~\ref{table:pcadp1bit}. Suppose that it holds for all $\alpha, \beta, \gamma \in \zspc{n}$, $a \in \zspc{}$. We prove that it is true for $\alpha', \beta', \gamma' \in \zspc{n+1}$, $a' \in \zspc{}$.
    Let $p = (\alpha'_n, \beta'_n, \gamma'_n)$ and $\alpha' = \alpha p_0, \beta' = \beta p_1$, $\gamma' = \gamma p_2$. Theorem~\ref{th:pcadprecurrence} implies that the inequality holds if $\wt{p}$ is odd and provides the following if $\wt{p}$ is even: 
    \begin{align}\label{eq:padpsumrecurrence}
        \padp{a}{0}{\alpha'}{\beta'}{\gamma'} + \padp{\overline{a}}{1}{\alpha'}{\beta'&}{\gamma'} \nonumber\\
        = \frac{1}{2^{\wt{p}}}\sum\limits_{q \in \zspc{3}, q \preceq p} \big(&\padp{a\oplus q_0}{q_1}{\inv{\alpha}{q_0}}{\inv{\beta}{q_1}}{\inv{\gamma}{q_2}} \nonumber\\
        + \ &\padp{\overline{a \oplus q_0}}{\overline{q_1}}{\inv{\alpha}{q_0}}{\inv{\beta}{q_1}}{\inv{\gamma}{q_2}}\big).
    \end{align}
    
    \textbf{Case 1.} $p_0 = 0$. In this case $\padp{0}{0}{\alpha'}{\alpha'}{0} + \padp{1}{1}{\alpha'}{\alpha'}{0} = \padp{0}{0}{\alpha}{\alpha}{0} + \padp{1}{1}{\alpha}{\alpha}{0}$. If $(p_1, p_2) = 0$, the equality~(\ref{eq:padpsumrecurrence}) and the induction hypothesis proves the statement. Let $(p_1, p_2) = (1, 1)$. Then
    \begin{multline*}
        \padp{a}{0}{\alpha'}{\beta'}{\gamma'} + \padp{\overline{a}}{1}{\alpha'}{\beta'}{\gamma'} \\
        = \frac{1}{4}\sum\limits_{(q_1, q_2) \in \zspc{2}} (\padp{a}{q_1}{\alpha}{\inv{\beta}{q_1}}{\inv{\gamma}{q_2}}
        + \padp{\overline{a}}{\overline{q_1}}{\alpha}{\inv{\beta}{q_1}}{\inv{\gamma}{q_2}}).
    \end{multline*}
    At the same time,
    \begin{multline*}
    	\padp{a}{q_1}{\alpha}{\inv{\beta}{q_1}}{\inv{\gamma}{q_2}} + \padp{\overline{a}}{\overline{q_1}}{\alpha}{\inv{\beta}{q_1}}{\inv{\gamma}{q_2}}\\
    	 \leq \padp{0}{0}{\alpha}{\alpha}{0} + \padp{1}{1}{\alpha}{\alpha}{0}
    \end{multline*}
    by the induction hypothesis. 

    \textbf{Case 2.} $p_0 = 1$. First of all, we note that the equality~(\ref{eq:padpsumrecurrence}) consists of $4$ terms. Moreover, two of them are zero, since the last bits of their arguments are of odd weight (see Corollary~\ref{cor:AdpMatrixZeros}). Next, 
    \begin{align}
        \padp{0}{0}{\alpha'}{\alpha'}{0} &+ \padp{1}{1}{\alpha'}{\alpha'}{0} \nonumber\\
        &= \frac{1}{4}(\padp{0}{0}{\alpha}{\alpha}{0} + \padp{1}{1}{\alpha}{\alpha}{0})\label{eq:padpsumaa0}\\
        &+ \frac{1}{4}(\padp{1}{1}{\overline{\alpha}}{\overline{\alpha}}{0} + \padp{0}{0}{\overline{\alpha}}{\overline{\alpha}}{0}).\label{eq:padpsumnana0}
    \end{align}
    Let $(p_1, p_2) = (1, 0)$. According to the equality~(\ref{eq:padpsumrecurrence}),
    \begin{align}
        \padp{a}{0}{\alpha'}{\beta'}{&\gamma'} + \padp{\overline{a}}{1}{\alpha'}{\beta'}{\gamma'} \nonumber\\
        &= \frac{1}{4}\sum\limits_{q \in \zspc{}} (\padp{a}{q}{\alpha}{\inv{\beta}{q}}{\gamma}
        + \padp{\overline{a}}{\overline{q}}{\alpha}{\inv{\beta}{q}}{\gamma}) \label{eq:padpsumabc}\\
        &+ \frac{1}{4}\sum\limits_{q \in \zspc{}} (\padp{\overline{a}}{q}{\overline{\alpha}}{\inv{\beta}{q}}{\gamma}
        + \padp{a}{\overline{q}}{\overline{\alpha}}{\inv{\beta}{q}}{\gamma}).\label{eq:padpsumnabc}
    \end{align}
    Note that in both (\ref{eq:padpsumabc}) and (\ref{eq:padpsumnabc}) a term for $q = 0$ or for $q = 1$ equals to zero. Therefore, the induction hypothesis provides that (\ref{eq:padpsumabc}) $\leq$ (\ref{eq:padpsumaa0}) and (\ref{eq:padpsumnabc}) $\leq$ (\ref{eq:padpsumnana0}).
    The case of $(p_1, p_2) = (0, 1)$ is the same.
\end{proof}
Finally, the next theorem describes how to find $\max_{\beta, \gamma}  \adpxr{\alpha}{\beta}{\gamma}{1}$.
\begin{theorem}\label{th:max1left}
        Let us fix the first argument of $\adpxrsym$. Then
    $$
    \underset{\beta, \gamma \in \zspc{n}}{\mathrm{max}} \ \adpxr{\alpha}{\beta}{\gamma}{1} = \adpxr{\alpha}{\alpha}{0}{1} \text{, where } \alpha \in \zspc{n}.
    $$
\end{theorem}
\begin{proof}
    Let $\alpha, \beta, \gamma \in \zspc{n - 1}$, $\alpha', \beta', \gamma' \in \zspc{}$, $a = \alpha_{n - 2} \oplus \beta_{n - 2} \oplus \gamma_{n - 2}$ and $a' = \alpha' \oplus \beta' \oplus \gamma'$. According to Corollary~\ref{cor:adrxrformula}, 
    $$\adpxr{\cctn{\alpha'}{\alpha}}{\cctn{\beta'}{\beta}}{\cctn{\gamma}{\gamma'}}{1} = p \cdot \padp{a'}{0}{\alpha}{\beta}{\inv{\gamma}{a}}
     + q \cdot \padp{\overline{a'}}{1}{\alpha}{\beta}{\inv{\gamma}{a}},$$
    where $p = \cadp{a}{\inv{\alpha'}{a'}}{\beta'}{\gamma'} \leq 1$ and $q = \cadp{a}{\inv{\overline{\alpha'}}{a'}}{\overline{\beta'}}{\gamma'} \leq 1$.
    At the same time, 
    \begin{align*}
     \adpxr{\cctn{\alpha'}{\alpha}}{\cctn{\alpha'}{\alpha}}{0}{1} = \padp{0}{0}{\alpha}{\alpha}{0} + \padp{1}{1}{\alpha}{\alpha}{0}.
    \end{align*}    
    Since $p,q \leq 1$, Lemma~\ref{lemma:padpsum} provides that $p \cdot \padp{a'}{0}{\alpha}{\beta}{\inv{\gamma}{a}}
     + q \cdot \padp{\overline{a'}}{1}{\alpha}{\beta}{\inv{\gamma}{a}} \leq \adpxr{\cctn{\alpha'}{\alpha}}{\cctn{\alpha'}{\alpha}}{0}{1}$.
\end{proof}
\begin{remark}
    The argument symmetries from Theorem~\ref{th:symmetries} guarantee that there are at least two optimal differentials if $\alpha \in \zspc{n} \setminus \{0, 2^{n - 1}\}$. Indeed,
    $$
        \underset{\beta, \gamma \in \zspc{n}}{\mathrm{max}} \ \adpxr{\alpha}{\beta}{\gamma}{1} = \adpxr{\alpha}{\alpha}{0}{1} = \adpxr{\alpha}{-\alpha}{0}{1}.
    $$ 
\end{remark} 

Note that the values of $\adpxrsym$ and $\adpsym$ coincide for the following differences.
\begin{proposition}\label{prop:adpxr_x_g0}
Let $\alpha, \beta \in \zspc{n}$. Then $\adpxr{\alpha}{\beta}{0}{r} = \adp{\alpha}{\beta}{0}$.
\end{proposition}
\begin{proof}
    By definition:
    \begin{align*}
        \adp{\alpha}{\beta}{0} &= 4^{-n} \#\{x, y \in \zspc{n} : (x + \alpha) \oplus (y + \beta) = x \oplus y \} \\
        &= 4^{-n} \#\{x, y \in \zspc{n} : ((x + \alpha) \oplus (y + \beta)) \lll r = (x \oplus y) \lll r \} \\
        &= \adpxr{\alpha}{\beta}{0}{r}. 
    \end{align*}
\end{proof}
 This means that there is the connection between maximums of $\adpsym$ and $\adpxrsym$.
\begin{corollary}\label{cor:onebitleft_adp}
    Let $\alpha \in \zspc{n}$. Then
    $$
        \underset{\beta, \gamma \in \zspc{n}}{\mathrm{max}} \ \adpxr{\alpha}{\beta}{\gamma}{1} = \adp{\alpha}{\alpha}{0} = \underset{\beta, \gamma \in \zspc{n}}{\mathrm{max}} \ \adp{\alpha}{\beta}{\gamma}.
    $$
\end{corollary}
\noindent The proof directly follows from Theorem~\ref{th:max1left}, Proposition~\ref{prop:adpxr_x_g0} and \cite[Theorem~2]{MouhaEtAl2021}. 

This property shows us the similarity between $\adpsym$ and $\adpxrsym$ if the rotation is one bit left. Overall, it is the most similar to $\adpsym$ case. This will also be confirmed in 
Section~\ref{sec:impossibledifferentialsest}.

Note that it is still not easy to calculate the exact value of $\max_{\beta, \gamma \in \zspc{n}}\adpxr{\alpha}{\beta}{\gamma}{1}$ without Theorem~\ref{th:AdpMatrix}/Theorem~\ref{th:adrxrformula}/Corollary~\ref{cor:adrxrformula}. According to~\cite[Proposition 6]{MouhaEtAl2021}, we can do this in the way similar to Theorem~\ref{th:AdpMatrix} with $2 \times 2$ matrices instead of $8 \times 8$. 
However, $\min_{\alpha \in \zspc{n}}\adp{\alpha}{\alpha}{0}$, i.e. $\min_{\alpha \in \zspc{n}} \max_{\beta, \gamma \in \zspc{n}}\adpxr{\alpha}{\beta}{\gamma}{1}$, that was studied in~\cite[Theorem 4]{MouhaEtAl2021}, can be calculated similarly to Fibonacci numbers.
A problem related to this minimum value is also presented in~\cite{NSUCRYPTO2021}.

\subsection{Maximums of \texorpdfstring{$\adpxrsym$}{} for \texorpdfstring{$r=n-1$}{}}\label{sec:maxfor1right}

Let us start with the following lemmas.
\begin{lemma}\label{lemma:cadp0max}
    Let $\alpha \in \zspc{n}$ and $c \in \zspc{}$. Then it is true that $\underset{\beta, \gamma \in \zspc{n}}{\mathrm{max}} \, \cadp{c}{\alpha}{\beta}{\gamma} \leq \underset{\beta, \gamma \in \zspc{n}}{\mathrm{max}} \, \cadp{0}{\alpha}{\beta}{\gamma} 
        = \underset{\beta, \gamma \in \zspc{n}}{\mathrm{max}} \, \adp{\alpha}{\beta}{\gamma}         = \adp{\alpha}{\alpha}{0}.$
\end{lemma}
\begin{proof}
    It is clear that $\max_{\beta,\gamma} \cadp{c}{\alpha}{\beta}{\gamma} \leq \max_{\beta,\gamma} \adp{\alpha}{\beta}{\gamma}$. Also, \cite[Theorem 2]{MouhaEtAl2021} provides that $\max_{\beta,\gamma} \adp{\alpha}{\beta}{\gamma} = \adp{\alpha}{\alpha}{0}$. At the same time, $\adp{\alpha}{\alpha}{0} = \cadp{0}{\alpha}{\alpha}{0} + \cadp{1}{\alpha}{\alpha}{0} = \cadp{0}{\alpha}{\alpha}{0}$ by Theorem~\ref{th:padpsymmetries}.
\end{proof}
\begin{lemma}\label{lemma:adpanamin}
    Let $\alpha \in \zspc{n}$, $\alpha_{n - 1} = 0$. Then $\adp{\overline{\alpha}}{\overline{\alpha}}{0} \leq \adp{\alpha}{\alpha}{0}$.
\end{lemma}
\begin{proof}
The case of $n = 1$ is straightforward. Next, $\adp{\alpha'1}{\alpha'1}{0} < \adp{\alpha'0}{\alpha'0}{0}$ holds for any $\alpha' \in \zspc{n - 1}$ by~\cite[Corollary 5]{MouhaEtAl2021}. Let $\alpha = \alpha'0$. Due to the symmetries of $\adpsym$ (see \cite[Proposition 3]{MouhaEtAl2021}), $\adp{\alpha'1}{\alpha'1}{0} = \adp{-(\alpha'1)}{-(\alpha'1)}{0}$. Also, $-(\alpha'1) = \overline{\alpha'}1 = \overline{\alpha'0} = \overline{\alpha}$.
\end{proof}

The next theorem describes how to find $\max_{\beta, \gamma}\adpxr{\alpha}{\beta}{\gamma}{n-1}$.
\begin{theorem}\label{th:max1right}
        Let us fix the first argument of $\adpxrsym$. Then 
    \begin{enumerate}
        \item $ \underset{\beta, \gamma \in \zspc{n}}{\mathrm{max}} \ \adpxr{\alpha0}{\beta}{\gamma}{n-1} = \adpxr{\alpha0}{\alpha0}{0}{n-1},$ $\alpha \in \zspc{n - 1}$,
        \item $ \underset{\beta, \gamma \in \zspc{n}}{\mathrm{max}} \ \adpxr{\alpha01}{\beta}{\gamma}{n-1} = \adpxr{\alpha01}{\alpha00}{2^{n-1}}{n-1},$ $\alpha \in \zspc{n - 2}$,
        \item $ \underset{\beta, \gamma \in \zspc{n}}{\mathrm{max}} \ \adpxr{\alpha11}{\beta}{\gamma}{n-1} = \adpxr{\alpha11}{\overline{\alpha}00}{2^{n-1}}{n-1},$ $\alpha \in \zspc{n - 2}$,
    \end{enumerate}
    where the vector $(\alpha_0, \ldots, \alpha_{n - 3}, a, b)$ is denoted by $\alpha a b$ for $\alpha \in \zspc{n-2}$, $a, b \in \zspc{}$. 
    \end{theorem}
\begin{proof}
 According to Corollary~\ref{cor:adrxrformula}, $\adpxr{\cctn{\alpha'}{\alpha}}{\cctn{\beta'}{\beta}}{\cctn{\gamma}{\gamma'}}{n - 1} = p \cdot \cadp{a}{\inv{\alpha'}{a'}}{\beta'}{\gamma'} + q \cdot \cadp{a}{\inv{\overline{\alpha'}}{a'}}{\overline{\beta'}}{\gamma'}$,
where $p = \padp{a'}{0}{\alpha}{\beta}{\inv{\gamma}{a}}$, $q = \padp{\overline{a'}}{1}{\alpha}{\beta}{\inv{\gamma}{a}}$.
Since $\alpha', \beta', \gamma' \in \zspc{n - 1}$ and $\alpha, \beta, \gamma \in \zspc{}$, Table~\ref{table:pcadp1bit} gives us all possible values of $p$ and $q$:

    \begin{center}
    \begin{minipage}{23em}
    \begin{tabular}{@{}lcccccccc@{}}
        \toprule
        $\alpha$ & \multicolumn{4}{@{}c@{}}{$0$} & \multicolumn{4}{@{}c@{}}{$1$}\\
        \cmidrule(lr){2-5}\cmidrule(l){6-9}
        $\beta$ & \multicolumn{2}{@{}c@{}}{$0$} & \multicolumn{2}{@{}c@{}}{$1$} & \multicolumn{2}{@{}c@{}}{$0$} & \multicolumn{2}{@{}c@{}}{$1$}\\
        \cmidrule(lr){2-3}\cmidrule(lr){4-5}\cmidrule(lr){6-7}\cmidrule(l){8-9}
        $\alpha \beta \inv{\gamma}{a}$ & $000$ & $000$ & $011$ & $011$ & $101$ & $101$ & $110$ & $110$\\
        \addlinespace
        $a'$ & $0$ & $1$ & $0$ & $1$ & $0$ & $1$ & $0$ & $1$\\
        \midrule
        $p$ & $1$   & $0$ & $1/2$ &  $0$   & $1/2$ & $1/2$ & $1/4$ & $1/4$\\
        $q$ & $0$   & $0$ & $0$   &  $1/2$ & $0$   & $0$   & $1/4$ & $1/4$\\
        \bottomrule
    \end{tabular}
    \end{minipage}
    \end{center}

\textbf{Case 1.} $\alpha = 0$, i.e. $\cctn{\alpha'}{\alpha} = \alpha'0$.  First of all, \begin{align*}
     \adpxr{\alpha'0}{\alpha'0}{0}{n-1} \stackrel{a'= 0}= \cadp{0}{\alpha'}{\alpha'}{0} = \adp{\alpha'}{\alpha'}{0}.
\end{align*}    
According to the first four columns of the table above, $\adpxr{\alpha'0}{\beta'\beta}{\cctn{\gamma}{\gamma'}}{n-1}$ takes one the following values: $\cadp{a}{\alpha'}{\beta'}{\gamma'}$, $0$, $\frac{1}{2} \cadp{a}{\alpha'}{\beta'}{\gamma'}$ and $\frac{1}{2} \cadp{a}{\overline{\overline{\alpha'}}}{\overline{\beta'}}{\gamma'}$.
In light of Lemma~\ref{lemma:cadp0max}, it is not difficult to see that any of them is not more than  $\adp{\alpha'}{\alpha'}{0}$. The first point is proved.

\textbf{Case 2.} $\alpha = 1$, i.e. $\cctn{\alpha'}{\alpha} = \alpha'1$. According to the last four columns of the table above, $\adpxr{\alpha'1}{\beta'\beta}{\cctn{\gamma}{\gamma'}}{n-1}$ takes one of the following values:
\begin{align}
       &\frac{1}{2}\cadp{a}{\alpha'}{\beta'}{\gamma'} \label{eq:tablecol5} \\
       &\frac{1}{2}\cadp{a}{\overline{\alpha'}}{\beta'}{\gamma'} \label{eq:tablecol6}\\
       &\frac{1}{4}\cadp{a}{\alpha'}{\beta'}{\gamma'} + \frac{1}{4}\cadp{a}{\overline{\alpha'}}{\overline{\beta'}}{\gamma'} \label{eq:tablecol7}\\
       &\frac{1}{4}\cadp{a}{\overline{\alpha'}}{\beta'}{\gamma'} + \frac{1}{4}\cadp{a}{\alpha'}{\overline{\beta'}}{\gamma'} \label{eq:tablecol8}
\end{align}

\textbf{Case 2.1} $\alpha'_{n - 2} = 0$. According to the second point of the theorem, let us define
$
  m_0 = \adpxr{\alpha'1}{\alpha'0}{2^{n-1}}{n-1} \stackrel{a'= 0}{=} 
  \frac{1}{2} \cadp{0}{\alpha'}{\alpha'}{0} = \frac{1}{2} \adp{\alpha'}{\alpha'}{0}.
$
Then by Lemmas~\ref{lemma:cadp0max} and \ref{lemma:adpanamin} (they are marked bellow as \ref{lemma:cadp0max} and \ref{lemma:adpanamin} respectively) we obtain that
\begin{align*}
    m_0 &= \frac{1}{2}\adp{\alpha'}{\alpha'}{0} \stackrel{\ref{lemma:cadp0max}}{\geq} (\ref{eq:tablecol5}),\\
    m_0 &\stackrel{\ref{lemma:adpanamin}}{\geq} \frac{1}{2}\adp{\overline{\alpha'}}{\overline{\alpha'}}{0} \stackrel{\ref{lemma:cadp0max}}{\geq}  (\ref{eq:tablecol6}),\\
    m_0 &= \frac{1}{4}\adp{\alpha'}{\alpha'}{0} + \frac{1}{4}\adp{\alpha'}{\alpha'}{0}\\ &\stackrel{\ref{lemma:adpanamin}}{\geq} \frac{1}{4}\adp{\alpha'}{\alpha'}{0} + \frac{1}{4}\adp{\overline{\alpha'}}{\overline{\alpha'}}{0} \stackrel{\ref{lemma:cadp0max}}{\geq}  (\ref{eq:tablecol7}),  (\ref{eq:tablecol8}).
\end{align*}
The second point is proved.

\textbf{Case 2.2} $\alpha'_{n - 2} = 1$. According to the third point of the theorem, let us define
$
    m_1 = \adpxr{\alpha'1}{\overline{\alpha'}0}{2^{n-1}}{n-1} \stackrel{a'= 1}{=} 
  \frac{1}{2} \cadp{0}{\overline{\alpha'}}{\overline{\alpha'}}{0} = \frac{1}{2} \adp{\overline{\alpha'}}{\overline{\alpha'}}{0}.
$
Similarly to the previous case, we obtain that
\begin{align*}
    m_1 &\stackrel{\ref{lemma:adpanamin}}{\geq} \frac{1}{2}\adp{\alpha'}{\alpha'}{0} \stackrel{\ref{lemma:cadp0max}}{\geq} (\ref{eq:tablecol5}),\\
    m_1 &= \frac{1}{2}\adp{\overline{\alpha'}}{\overline{\alpha'}}{0}
    \stackrel{\ref{lemma:cadp0max}}{\geq}  (\ref{eq:tablecol6}),\\
    m_1 &= \frac{1}{4}\adp{\overline{\alpha'}}{\overline{\alpha'}}{0} + \frac{1}{4}\adp{\overline{\alpha'}}{\overline{\alpha'}}{0}\\
    &\stackrel{\ref{lemma:adpanamin}}{\geq} \frac{1}{4}\adp{\alpha'}{\alpha'}{0} + \frac{1}{4}\adp{\overline{\alpha'}}{\overline{\alpha'}}{0}
    \stackrel{\ref{lemma:cadp0max}}{\geq}  (\ref{eq:tablecol7}),  (\ref{eq:tablecol8}).
\end{align*}
The theorem is proved.
\end{proof}
\begin{remark}
   The argument symmetries from Theorem~\ref{th:symmetries} guarantee that there are at least two optimal differentials if the first argument $\alpha' \in \zspc{n} \setminus \{0, 1, 2^{n - 1} - 1, 2^{n - 1}, 2^{n - 1} + 1, 2^{n} - 1\}$. Indeed, 
$$
    \max_{\beta, \gamma \in \zspc{n}}\adpxr{\alpha'}{\beta}{\gamma}{n - 1} = \adpxr{\alpha'}{\beta'}{\gamma'}{n - 1} = \adpxr{\alpha'}{-\beta'}{\gamma'}{n - 1}
$$
for $\beta'$ and $\gamma'$ given in the theorem. Also, $\beta' \neq -\beta'$ for the considered $\alpha'$.  Unfortunately, we cannot extend this symmetry since $\alpha'$ is fixed and $\gamma' = -\gamma'$.
\end{remark} 

Similarly to Corollary~\ref{cor:onebitleft_adp}, there is the connection with the maximums of $\adpsym$.
\begin{corollary}\label{cor:onebitright_adp}
    Let $\alpha \in \zspc{n-1}$. Then
    $$
        \underset{\beta, \gamma \in \zspc{n}}{\mathrm{max}} \ \adpxr{\alpha0}{\beta}{\gamma}{n-1} = \adp{\alpha0}{\alpha0}{0} = \underset{\beta, \gamma \in \zspc{n}}{\mathrm{max}} \ \adp{\alpha0}{\beta}{\gamma}.
    $$
\end{corollary}
\noindent The proof directly follows from Theorem~\ref{th:max1right}, Proposition~\ref{prop:adpxr_x_g0} and \cite[Theorem~2]{MouhaEtAl2021}.
Thus, in a half cases the considered maximums of $\adpsym$ and $\adpxrsym$ coincide. 

\section{Impossible differentials of the XR-operation}\label{sec:inconsistentdiffxr}

In this section, we describe all arguments such that $\adpxrsym$ takes $0$ on them, i.e. all impossible differentials of the function $(x \oplus y) \lll r$. The section is divided into three parts. Section~\ref{sec:patterns} defines patterns for octal words that are used to describe differentials. In Section~\ref{sec:impossibledifferentialsdescr}, all impossible differentials of the XR-operation 
are described. Finally, Section~\ref{sec:impossibledifferentialsest} provides some estimations for the number of the all impossible differentials. Note that the knowledge of all impossible differentials may be helpful for constructing impossible differential attacks.

\subsection{Patterns for octal words}\label{sec:patterns}

Let $\alpha, \beta, \gamma \in \zspc{n}$. To classify octal words $\omega = \omega(\alpha, \beta, \gamma)$, we use patterns that are similar to regular expressions. Pattern elements are presented in Table~\ref{table:octalwordpatterns}. The word $\omega$ satisfies a pattern if all octal symbols of the word satisfy corresponding pattern symbols. Let us consider several examples.
\begin{itemize}
    \item The word $\omega$ satisfies \texttt{[.*e]} $\Longleftrightarrow$ the least significant octal symbol $\omega_{n - 1}$ is $0, 3, 5$ or $6$. In other words, $\omega_{n - 1}$ is of even weight or, equivalently, $\alpha_{n - 1} \oplus \beta_{n - 1} \oplus \gamma_{n - 1} = 0$. All other $\omega_{0}, \ldots, \omega_{n - 2}$ are arbitrary.
        
    \item The word $\omega$ satisfies \texttt{[.*d0*]} $\Longleftrightarrow$ the first nonzero octal symbol of $\omega$ (starting with the least significant bits) exists and belongs to $\{1, 2, 4, 7\}$, i.e. it is of odd weight.

\end{itemize}

\begin{table}[h]
\begin{center}
\begin{tabular}{@{}ll@{}}
        \toprule
        Pattern symbol & Octal word elements\\
        \midrule
        \texttt{[} & the begin (most significant bits)\\
        \texttt{]} & the end (least significant bits)\\
        \texttt{.} & $0, 1, \ldots, 7$ (any octal symbol)\\
        \texttt{e} and \texttt{d} & $0, 3, 5, 6$ and $1, 2, 4, 7$ resp.\\
        \texttt{$\widehat{\mathtt{t}}$} & $t$, $t \oplus 3$, $t \oplus 5$\\
        \texttt{$\mathtt{\alpha_0}$} and \texttt{$\mathtt{\alpha_1}$}& $0, 1, 2, 3$ and $4, 5, 6, 7$ respectively\\
        \texttt{$\mathtt{\beta_0}$} and \texttt{$\mathtt{\beta_1}$} & $0, 1, 4, 5$ and $2, 3, 6, 7$ respectively\\
        \texttt{$\mathtt{\gamma_0}$} and \texttt{$\mathtt{\gamma_1}$} & $0, 2, 4, 6$ and $1, 3, 5, 7$ respectively\\
        \texttt{$\mathtt{0_6}$} and \texttt{$\mathtt{1_7}$} & $0, 6$ and $1, 7$ respectively\\
        \texttt{s*} & no symbols, \texttt{s}, \texttt{ss}, \texttt{sss}, \texttt{ssss}, \ldots \\
        \bottomrule
    \end{tabular}
\caption{Pattern elements}
\label{table:octalwordpatterns}
\end{center}
\end{table}    

For instance, Corollary~\ref{cor:AdpMatrixZeros} is more compact in terms of patterns.
\begin{proposition}[Lipmaa et al. \cite{LipmaaEtAl2004}]\label{prop:adpzeropattern}
    Let $\alpha, \beta, \gamma \in \zspc{n}$. Then $\adp{\alpha}{\beta}{\gamma} = 0$ $\iff$ $\omega(\alpha, \beta, \gamma)$ satisfies {\rm\texttt{[.*d0*]}}.
\end{proposition}

Let us define symmetries for patterns.
\begin{enumerate}
    \item A pattern $p_1$ is $\alpha\beta$-symmetric to a pattern $p_2$ if for any $\alpha, \beta, \gamma \in \zspc{n}$ the following holds: $\omega(\alpha, \beta, \gamma)$ satisfies $p_1$ $\Longleftrightarrow$ $\omega(\beta, \alpha, \gamma)$ satisfies $p_2$;
    \item $p_1$ is $\overline{\gamma}$-symmetric to $p_2$ if for any $\alpha, \beta, \gamma \in \zspc{n}$ the following holds: $\omega(\alpha, \beta, \gamma)$ satisfies $p_1$ $\Longleftrightarrow$ $\omega(\alpha, \beta, \overline{\gamma})$ satisfies $p_2$. 
\end{enumerate}
Also, a pattern $p_1$ satisfies a pattern $p_2$ if any word satisfying $p_1$ satisfies $p_2$ as well.

\subsection{The characterization of the impossible differentials}\label{sec:impossibledifferentialsdescr}

We start with determining if auxiliary $\cadppar{c}$ and $\padppar{a}{b}$ are zero.
\begin{proposition}\label{prop:cadpzeros}
    Let $\alpha, \beta, \gamma \in \zspc{n}$. Then
    \begin{enumerate}
        \item $\cadp{0}{\alpha}{\beta}{\gamma} \neq 0$ $\iff$ $\adp{\alpha}{\beta}{\gamma} \neq 0$;
        \item $\cadp{1}{\alpha}{\beta}{\gamma} \neq 0$ $\iff$ $\adp{\alpha}{\beta}{\gamma} \neq 0$ and $\gamma \neq 0$.
    \end{enumerate}
\end{proposition}
\begin{proof}
    By induction: if $n = 1$, Table~\ref{table:pcadp1bit} provides the base of the induction.
    We suppose that the proposition holds for $n$. Let us show that the statement holds for $n + 1$.
    Let $\alpha, \beta, \gamma \in \zspc{n}, p \in \zspc{3}, c \in \zspc{}$. According to Theorem~\ref{th:pcadprecurrence}, 
    \begin{align*}
        \cadp{c}{\alpha p_0}{\beta p_1}{\gamma p_2} &=
            \frac{1}{2^{\wt{p}}}\sum\limits_{q \in \zspc{3}, q \preceq p}\cadp{c \oplus q_2}{\inv{\alpha}{q_0}}{\inv{\beta}{q_1}}{\inv{\gamma}{q_2}}
    \end{align*}
    if $\wt{p}$ is even. Also, $\cadp{c}{\alpha p_0}{\beta p_1}{\gamma p_2} = 0$ if $\wt{p}$ is odd and, similarly, $\adp{\alpha p_0}{\beta p_1}{\gamma p_2} = 0$ by Corollary~\ref{cor:AdpMatrixZeros}. Thus, we need only to consider the case of even $\wt{p}$.

    \textbf{Case 1.} $p = (0, 0, 0)$. According to the formula above, $\cadp{c}{\alpha p_0}{\beta p_1}{\gamma p_2} = \cadp{c}{\alpha}{\beta}{\gamma}$. The induction hypothesis and Corollary~\ref{cor:AdpMatrixZeros} prove the induction step.
    
    For all other cases $\adp{\alpha p_0}{\beta p_1}{\gamma p_2} > 0$ by Corollary~\ref{cor:AdpMatrixZeros}.
    
    \textbf{Case 2.} $p = \{(0, 1, 1), (1, 0, 1)\}$. Exactly one of $\{(\alpha_{n - 1}, \beta_{n - 1}, \gamma_{n - 1}) \oplus q : q \preceq p \text{ and } \gamma_{n - 1} \oplus q_2 = 1\}$ is a nonzero vector of even weight. Choosing corresponding $q$, we obtain that $\adp{\inv{\alpha}{q_0}}{\inv{\beta}{q_1}}{\inv{\gamma}{q_2}} > 0$ by Corollary~\ref{cor:AdpMatrixZeros}. By induction, $\cadp{c \oplus q_2}{\inv{\alpha}{q_0}}{\inv{\beta}{q_1}}{\inv{\gamma}{q_2}} > 0$ as well since $\inv{\gamma}{q_2} \neq 0$. Thus, $\cadp{c}{\alpha p_0}{\beta p_1}{\gamma p_2} > 0$.  
    
    \textbf{Case 3.} $p = (1, 1, 0)$, i.e. $q_2 = 0$ if $q \preceq p$. At least one of $\{(\alpha_{n - 1}, \beta_{n - 1}, \gamma_{n - 1}) \oplus q : q \preceq p \}$ is a nonzero vector of even weight. Choosing corresponding $q$, we obtain that $\adp{\inv{\alpha}{q_0}}{\inv{\beta}{q_1}}{\inv{\gamma}{q_2}} > 0$ by Corollary~\ref{cor:AdpMatrixZeros}. 
    Hence, $\cadp{0}{\alpha p_0}{\beta p_1}{\gamma p_2} > 0$ by induction. 
    If $\gamma \neq 0$, $\cadp{1}{\alpha p_0}{\beta p_1}{\gamma p_2} > 0$ as well.
    If $\gamma = 0$, i.e. $\gamma p_2 = 0$, then $\cadp{1}{\inv{\alpha}{q_0}}{\inv{\beta}{q_1}}{\inv{\gamma}{q_2}} = \cadp{1}{\inv{\alpha}{q_0}}{\inv{\beta}{q_1}}{\gamma} = 0$ for all $q \preceq p$ by induction. As a result, $\cadp{1}{\alpha p_0}{\beta p_1}{\gamma p_2} = 0$. 
\end{proof}

However, it is not easy to determine if $\padppar{a}{b}$ is zero.
\begin{theorem}\label{th:padpzeros}
    Let $\alpha, \beta, \gamma \in \zspc{n}$ and $a, b \in \zspc{}$. Then $\padp{a}{b}{\alpha}{\beta}{\gamma} = 0$ $\iff$ $\omega(\alpha, \beta, \gamma)$ satisfies some pattern given in Table~\ref{table:padpzeros} in the column marked as $\padppar{a}{b}$.
\end{theorem}
Its proof is based on the rational series approach and can be found in Appendix~\ref{sec:proofofpadpzeros}. Now we are ready to determine if $\adpxrsym$ is zero.
\begin{table}
\begin{center}
\begin{tabular}{@{}cllll@{}}
	\toprule
	 & \multicolumn{1}{@{}c@{}}{$\padppar{0}{0}$} & \multicolumn{1}{@{}c@{}}{$\padppar{0}{1}$} & \multicolumn{1}{@{}c@{}}{$\padppar{1}{0}$} & \multicolumn{1}{@{}c@{}}{$\padppar{1}{1}$} \\
	\midrule
	1 & \multicolumn{4}{@{}c@{}}{\texttt{[.*d0*]}} \\
 \addlinespace
	2 &  \texttt{[6$\widehat{\mathtt{6}}$*7.*]} & \texttt{[4$\widehat{\mathtt{4}}$*5.*]} & \texttt{[2$\widehat{\mathtt{2}}$*3.*]} & \texttt{[0$\widehat{\mathtt{0}}$*]}\\
	3 &  \texttt{[7$\widehat{\mathtt{7}}$*6.*]} & \texttt{[4$\widehat{\mathtt{4}}$*0$\mathtt{\beta_0}$*]} & \texttt{[2$\widehat{\mathtt{2}}$*0$\mathtt{\alpha_0}$*]} & \texttt{[0$\widehat{\mathtt{0}}$*1.*]} \\
	4 &  \texttt{[7$\widehat{\mathtt{7}}$*0$\mathtt{\gamma_0}$*]} & \texttt{[5$\widehat{\mathtt{5}}$*]} & \texttt{[3$\widehat{\mathtt{3}}$*]} & \texttt{[0$\widehat{\mathtt{0}}$*2$\mathtt{\alpha_0}$*]} \\
	5 &  & \texttt{[5$\widehat{\mathtt{5}}$*4.*]} & \texttt{[3$\widehat{\mathtt{3}}$*2.*]} & \texttt{[0$\widehat{\mathtt{0}}$*4$\mathtt{\beta_0}$*]} \\
	6 &  & \texttt{[5$\widehat{\mathtt{5}}$*1$\mathtt{\beta_0}$*]} & \texttt{[3$\widehat{\mathtt{3}}$*1$\mathtt{\alpha_0}$*]} & \texttt{[1$\widehat{\mathtt{1}}$*0.*]} \\
	7 &  & \texttt{[5$\widehat{\mathtt{5}}$*2$\mathtt{\gamma_0}$*]} & \texttt{[3$\widehat{\mathtt{3}}$*4$\mathtt{\gamma_0}$*]} & \texttt{[1$\widehat{\mathtt{1}}$*3$\mathtt{\alpha_0}$*]} \\
	8 &  &  \texttt{[$\mathtt{\beta_0}$*]} &  \texttt{[$\mathtt{\alpha_0}$*]} & \texttt{[1$\widehat{\mathtt{1}}$*5$\mathtt{\beta_0}$*]} \\
	9 &  &  &  & \texttt{[1$\widehat{\mathtt{1}}$*6$\mathtt{\gamma_0}$*]} \\
	10 &  &  &  & \texttt{[$\mathtt{\alpha_0}$*]} \\
	11 &  &  &  & \texttt{[$\mathtt{\beta_0}$*]}  \\
	\bottomrule
\end{tabular}
\caption{Zero values of $\padppar{a}{b}$}
\label{table:padpzeros}
\end{center}
\end{table}

\begin{theorem}\label{th:adpXRzeros}
	Let $\alpha, \beta, \gamma \in \zspc{n -r}$, $\alpha', \beta', \gamma' \in \zspc{r}$, $\omega = \omega(\alpha, \beta, \gamma)$ and $\omega' = \omega(\alpha', \beta', \gamma')$. Then $\adpxr{\cctn{\alpha'}{\alpha}}{\cctn{\beta'}{\beta}}{\cctn{\gamma}{\gamma'}}{r} = 0$ $\iff$ $(\omega, \omega')$ satisfies some pattern X.Y from Table~\ref{table:zerosOfXR}, i.e. $\omega'$ satisfies the mark of column X and $\omega$ satisfies the element of row Y and column X. 
\end{theorem}

\begin{table}
\begin{center}
\begin{minipage}{\linewidth}
{\footnotesize
\begin{tabular*}{\linewidth}{@{\extracolsep{\fill}}clllllll@{\extracolsep{\fill}}}
	\toprule
	 & \multicolumn{1}{@{}c@{}}{1} & \multicolumn{1}{@{}c@{}}{2} & \multicolumn{1}{@{}c@{}}{3} & \multicolumn{1}{@{}c@{}}{4} & \multicolumn{1}{@{}c@{}}{5} & \multicolumn{1}{@{}c@{}}{6} & \multicolumn{1}{@{}c@{}}{7}\\
	 & \multicolumn{1}{@{}c@{}}{\texttt{[.*d00*]}} & \multicolumn{1}{@{}c@{}}{\texttt{[.*e22*]}} & \multicolumn{1}{@{}c@{}}{\texttt{[.*e44*]}} & \multicolumn{1}{@{}c@{}}{\texttt{[.*d66*]}} & \multicolumn{1}{@{}c@{}}{\texttt{[.*d]}} & \multicolumn{1}{@{}c@{}}{\texttt{[.*]}} & \multicolumn{1}{@{}c@{}}{\texttt{[$\gamma_0$*]}}\\
	\midrule
	1 & \texttt{[$\alpha_0$*]} & \texttt{[$\alpha_0$*]} & \texttt{[$\beta_0$*]} & \texttt{[6$\widehat{\mathtt{6}}$*7.*]} & \texttt{[0*]} & \texttt{[.*d00*]} & \texttt{[.*d]} \\
	2 & \texttt{[$\beta_0$*]} & \texttt{[2$\widehat{\mathtt{2}}$*]} & \texttt{[4$\widehat{\mathtt{4}}$*]} & \texttt{[6$\widehat{\mathtt{6}}$*1$\gamma_1$*$\mathtt{1_7}$]} & \texttt{[1*]} &  \texttt{[.*e11*]} & \\
	3 & \texttt{[0$\widehat{\mathtt{0}}$*]} & \texttt{[2$\widehat{\mathtt{2}}$*3.*]} & \texttt{[4$\widehat{\mathtt{4}}$*5.*]} & \texttt{[7$\widehat{\mathtt{7}}$*6.*]} & & & \\
	4 & \texttt{[0$\widehat{\mathtt{0}}$*1.*]} & \texttt{[2$\widehat{\mathtt{2}}$*0$\alpha_0$*]} & \texttt{[4$\widehat{\mathtt{4}}$*0$\beta_0$*]} & \texttt{[7$\widehat{\mathtt{7}}$*0$\gamma_0$*$\mathtt{0_6}$]} &  &  & \\
	5 & \texttt{[0$\widehat{\mathtt{0}}$*2$\alpha_0$*]} & \texttt{[2$\widehat{\mathtt{2}}$*5$\gamma_1$*$\mathtt{1_7}$]} & \texttt{[4$\widehat{\mathtt{4}}$*3$\gamma_1$*$\mathtt{1_7}$]} &  &  &  & \\
	6 & \texttt{[0$\widehat{\mathtt{0}}$*4$\beta_0$*]} & \texttt{[3$\widehat{\mathtt{3}}$*]} & \texttt{[5$\widehat{\mathtt{5}}$*]} &  &  &  & \\
	7 & \texttt{[0$\widehat{\mathtt{0}}$*7]} & \texttt{[3$\widehat{\mathtt{3}}$*2.*]} & \texttt{[5$\widehat{\mathtt{5}}$*4.*]} &  &  &  & \\
	8 & \texttt{[0$\widehat{\mathtt{0}}$*7$\gamma_1$*$\mathtt{1_7}$]} & \texttt{[3$\widehat{\mathtt{3}}$*1$\alpha_0$*]} & \texttt{[5$\widehat{\mathtt{5}}$*1$\beta_0$*]} &  &  &  & \\
	9 & \texttt{[1$\widehat{\mathtt{1}}$*]} & \texttt{[3$\widehat{\mathtt{3}}$*4$\gamma_0$*$\mathtt{0_6}$]} & \texttt{[5$\widehat{\mathtt{5}}$*2$\gamma_0$*$\mathtt{0_6}$]}  &  &  &  & \\
	10 & \texttt{[1$\widehat{\mathtt{1}}$*0.*]} &  &  &  &  &  & \\
	11 & \texttt{[1$\widehat{\mathtt{1}}$*3$\alpha_0$*]} &  &  &  &  &  & \\
	12 & \texttt{[1$\widehat{\mathtt{1}}$*5$\beta_0$*]} &  &  &  &  &  & \\
	13 & \texttt{[1$\widehat{\mathtt{1}}$*6]} &  &  &  &  &  & \\
	14 & \texttt{[1$\widehat{\mathtt{1}}$*6$\gamma_0$*$\mathtt{0_6}$]} &  &  &  &  &  & \\
	\bottomrule
\end{tabular*}
}
\end{minipage}
\end{center}
\caption{Zero values of $\adpxrsym$}
\label{table:zerosOfXR}
\end{table}

\noindent For instance, $(\omega, \omega')$ satisfies 7.1 $\iff$ $\omega'$ satisfies \texttt{[$\mathtt{\gamma_0}$*]} and $\omega$ satisfies \texttt{[.*d]}. 
Also, $\omega$-parts of impossible differentials satisfy the following property: 
\begin{itemize}
    \item 1.1, 1.2, 2.1 and 3.1 are $\overline{\gamma}$-symmetric to itself;
    \item 1.3--1.8 are $\overline{\gamma}$-symmetric to 1.9--1.14;
    \item 2.2--2.5 are $\overline{\gamma}$-symmetric to 2.6--2.9;
    \item 3.2--3.5 are $\overline{\gamma}$-symmetric to 3.6--3.9; overall,  3.1--3.9 are $\alpha\beta$-symmetric to 2.1--2.9;
    \item 4.1--4.2 are $\overline{\gamma}$-symmetric to 4.3--4.4;
    \item 5.1 and 6.1 are $\overline{\gamma}$-symmetric to 5.2 and 6.2;
    \item 7.1 is the only case that is not $\overline{\gamma}$-symmetric to some other pattern.
\end{itemize}
\begin{proof}
    Let 
    \begin{align*}
        &a = \alpha_{n - r - 1} \oplus \beta_{n - r - 1} \oplus \gamma_{n - r - 1}, &&a' = \alpha'_{r - 1} \oplus \beta'_{r - 1} \oplus \gamma'_{r - 1},\\
        &x = \padp{a'}{0}{\alpha}{\beta}{\inv{\gamma}{a}}, &&y = \cadp{a}{\inv{\alpha'}{a'}}{\beta'}{\gamma'},\\
        &x' = \padp{\overline{a'}}{1}{\alpha}{\beta}{\inv{\gamma}{a}}, &&y' = \cadp{a}{\inv{\overline{\alpha'}}{a'}}{\overline{\beta'}}{\gamma'}.
    \end{align*}
    According to Corollary~\ref{cor:adrxrformula},
    \begin{align*}
     \adpxr{\cctn{\alpha'}{\alpha}}{\cctn{\beta'}{\beta}}{\cctn{\gamma}{\gamma'}}{r} &= x y + x' y'.
    \end{align*} 
    The value of $\adpxr{\cctn{\alpha'}{\alpha}}{\cctn{\beta'}{\beta}}{\cctn{\gamma}{\gamma'}}{r}$ is zero $\iff$ $x$, $x'$, $y$ and $y'$ satisfy any of the following conditions: $(y, y') = (0, 0)$; $(x, x') = (0, 0)$; $(x, y') = (0, 0)$; $(x', y) = (0, 0)$.

    \textbf{Case 1.} $y = 0$ and $y' = 0$. It is not difficult to see that both $\cadp{a}{\inv{\alpha'}{a'}}{\beta'}{\gamma'}$ and $\cadp{a}{\inv{\overline{\alpha'}}{a'}}{\overline{\beta'}}{\gamma'}$ are zero $\iff$ $\gamma' = 0$ and $a = 1$. Indeed, the least significant octal symbols of their arguments are even and at least one of them is not zero. This means that the corresponding $\adpsym$ is not zero due to Corollary~\ref{cor:AdpMatrixZeros}. Therefore, the corresponding $\cadppar{a}$ is zero $\iff$ $\gamma' = 0$ and $a = 1$ (see Proposition~\ref{prop:cadpzeros}). Thus, it corresponds to 7.1.

    \textbf{Case 2.} $x = 0$ and $x' = 0$. We divide this case into two subcases. 
    
    \textbf{Case 2.1.} $\adp{\alpha}{\beta}{\inv{\gamma}{a}} = 0$. Indeed, in this case both 
    $\padp{a'}{0}{\alpha}{\beta}{\inv{\gamma}{a}}$ and $\padp{\overline{a'}}{1}{\alpha}{\beta}{\inv{\gamma}{a}}$ are zero. Next, $\adp{\alpha}{\beta}{\inv{\gamma}{a}} = 0$ $\iff$ $\omega(\alpha, \beta, \inv{\gamma}{a})$ satisfies \texttt{[.*d0*]} (Proposition~\ref{prop:adpzeropattern}). But $\omega(\alpha, \beta, \inv{\gamma}{a})$ is an arbitrary word whose least significant symbol is of even weight. Hence, $\omega(\alpha, \beta, \inv{\gamma}{a})$ satisfies \texttt{[.*d00*]}. This means that $\omega(\alpha, \beta, \gamma)$ satisfies either 6.1 (i.e. $a = 0$) or 6.2 (i.e. $a = 1$).

    Next, we assume that $\adp{\alpha}{\beta}{\inv{\gamma}{a}} \neq 0$ for any other cases.

    \textbf{Case 2.2.} Let $a' = 0$. According to Theorem~\ref{th:padpzeros}, both $\padp{a'}{0}{\alpha}{\beta}{\inv{\gamma}{a}}$ and $\padp{\overline{a'}}{1}{\alpha}{\beta}{\inv{\gamma}{a}}$ cannot be zero. Indeed, $\padp{0}{0}{\alpha}{\beta}{\inv{\gamma}{a}} = 0$ only if the most significant octal symbol of $\omega(\alpha, \beta,\inv{\gamma}{a})$ is $6$ or $7$. Consequently, $\alpha \neq 0$ and $\beta \neq 0$, i.e. $\padp{1}{1}{\alpha}{\beta}{\inv{\gamma}{a}} \neq 0$. 

    Let $a' = 1$. According to Theorem~\ref{th:padpzeros}, both $\padp{1}{0}{\alpha}{\beta}{\inv{\gamma}{a}}$ and $\padp{0}{1}{\alpha}{\beta}{\inv{\gamma}{a}}$ are zero only if $\alpha = 0$ and $\beta = 0$. Indeed, in other cases $\padppar{0}{1}$ is zero only if the most significant octal symbol of $\omega(\alpha, \beta,\inv{\gamma}{a})$ is $5$ or $6$ which implies that $\alpha \neq 0$. Similarly, in other cases $\padppar{1}{0}$ is zero only if the symbol is $2$ or $3$ which implies that $\beta \neq 0$.
    Let $\alpha = 0$ and $\beta = 0$. Then $\inv{\gamma}{a} = 0$ as well (otherwise $\adp{\alpha}{\beta}{\inv{\gamma}{a}} = 0$ by Corollary~\ref{cor:AdpMatrixZeros} which is the case~2.1). Since both  $\padp{1}{0}{0}{0}{0}$ and $\padp{0}{1}{0}{0}{0}$ are zero (see Theorem~\ref{th:cadpsymmetries}), we obtain either $\alpha = \beta = \gamma = 0$ (i.e. $a = 0$) or $\alpha = \beta = 0$, $\gamma = 2^{n - r} - 1$ (i.e. $a = 1$). They correspond to 5.1 and 5.2.

    The patterns from the columns $5$, $6$ and $7$ are obtained.

    \textbf{Case 3.} $x = 0$ and $y' = 0$. Since $y' = \cadp{a}{\inv{\overline{\alpha'}}{a'}}{\overline{\beta'}}{\gamma'} = 0$, Proposition~\ref{prop:cadpzeros} gives us that  $\adp{\inv{\overline{\alpha'}}{a'}}{\overline{\beta'}}{\gamma'} = 0$ (we can exclude the case of $a = 1$ and $\gamma' = 0$ since it corresponds to 7.1). This means that $\omega(\inv{\overline{\alpha'}}{a'}, \overline{\beta'},\gamma')$ satisfies \texttt{[.*d00*]} (its least significant octal symbol is always of even weight, i.e. it must be zero). In other words, $\omega(\alpha', \beta', \gamma')$ satisfies either \texttt{[.*d66*]} (if $a' = 0$) or \texttt{[.*e22*]} (if $a' = 1$). Let us consider these two subcases.

    \textbf{Case 3.1.} $\omega(\alpha', \beta', \gamma')$ satisfies \texttt{[.*d66*]}, i.e. $a' = 0$. Thus, $x = \padp{0}{0}{\alpha}{\beta}{\inv{\gamma}{a}} = 0$ and $\omega(\alpha, \beta, \inv{\gamma}{a})$ must satisfy some pattern provided by Theorem~\ref{th:padpzeros}. For the rest of the proof we use the following.
    \begin{itemize}
        \item Since $\omega(\alpha, \beta, \inv{\gamma}{a})$ is an arbitrary word whose least significant symbol is of even weight, we have to choose only patterns satisfying both Table~\ref{table:padpzeros} (taking into account the column) and \texttt{[.*e]} (we can divide some patterns into several parts if it is necessary).
        \item After choosing a pattern for $\omega(\alpha, \beta, \inv{\gamma}{a})$, we choose it and its $\overline{\gamma}$-symmetric pattern as a pattern for $\omega(\alpha, \beta, \gamma)$ (it is already used in the previous cases).
        \item However, if Theorem~\ref{th:padpzeros} provides two patterns that are $\overline{\gamma}$-symmetric to each other (or some $\overline{\gamma}$-symmetric to itself pattern), we just choose them for $\omega(\alpha, \beta, \gamma)$ without any changes. Indeed, if a word satisfying a pattern does not satisfy \texttt{[.*e]}, then its $\overline{\gamma}$-symmetric word satisfies both \texttt{[.*e]} and its $\overline{\gamma}$-symmetric pattern. 
    \end{itemize}

    In this case, $\omega(\alpha, \beta, \inv{\gamma}{a})$ must satisfy \texttt{[6$\widehat{\mathtt{6}}$*7.*]}, \texttt{[7$\widehat{\mathtt{7}}$*6.*]}, or \texttt{[7$\widehat{\mathtt{7}}$*0$\mathtt{\gamma_0}$*]}. The first two are $\overline{\gamma}$-symmetric and correspond to 4.1 and 4.3. To satisfy \texttt{[.*e]}, we need to divide the last one into two patterns: \texttt{[7$\widehat{\mathtt{7}}$*0]} and \texttt{[7$\widehat{\mathtt{7}}$*0$\mathtt{\gamma_0}$*$\mathtt{0_6}$]}. 
    At the same time, \texttt{[7$\widehat{\mathtt{7}}$*0]} and its $\overline{\gamma}$-symmetric pattern together with $\omega(\alpha', \beta', \gamma')$ satisfy 6.1 and 6.2 respectively. 
    Finally, \texttt{[7$\widehat{\mathtt{7}}$*0$\mathtt{\gamma_0}$*$\mathtt{0_6}$]} and its $\overline{\gamma}$-symmetric pattern are 4.4 and 4.2. The fourth column of Table~\ref{table:zerosOfXR} is completed.

    \textbf{Case 3.2.} $\omega(\alpha', \beta', \gamma')$ satisfies \texttt{[.*e22*]}, i.e. $a' = 1$. Thus, $x = \padp{1}{0}{\alpha}{\beta}{\inv{\gamma}{a}} = 0$ and $\omega(\alpha, \beta, \inv{\gamma}{a})$ must satisfy one of the following patterns: 
    $\overline{\gamma}$-symmetric \texttt{[2$\widehat{\mathtt{2}}$*3.*]} and \texttt{[3$\widehat{\mathtt{3}}$*2.*]}, \texttt{[2$\widehat{\mathtt{2}}$*0$\mathtt{\alpha_0}$*]} and \texttt{[3$\widehat{\mathtt{3}}$*1$\mathtt{\alpha_0}$*]} (they correspond to 2.3, 2.4, 2.7 and 2.8); 
    $\overline{\gamma}$-symmetric to itself \texttt{[$\mathtt{\alpha_0}$*]} (2.1); two rest \texttt{[3$\widehat{\mathtt{3}}$*]} and \texttt{[3$\widehat{\mathtt{3}}$*4$\mathtt{\gamma_0}$*]}.

    The pattern \texttt{[3$\widehat{\mathtt{3}}$*]} satisfies \texttt{[.*e]}, it and its $\overline{\gamma}$-symmetric pattern correspond to 2.2 and 2.6. Also, we transform \texttt{[3$\widehat{\mathtt{3}}$*4$\mathtt{\gamma_0}$*]} to \texttt{[3$\widehat{\mathtt{3}}$*4$\mathtt{\gamma_0}$*$\mathtt{0_6}$]} (it and its $\overline{\gamma}$-symmetric pattern correspond to 2.5 and 2.9). We do not divide it into two patterns since the least significant symbol of \texttt{[3$\widehat{\mathtt{3}}$*4]} is of odd weight. The second column of Table~\ref{table:zerosOfXR} is completed.

    \textbf{Case 4.} $x' = 0$ and $y = 0$. Similarly to the previous case, we exclude the case of $a = 1$ and $\gamma' = 0$. Therefore, Proposition~\ref{prop:cadpzeros} gives us that $\adp{\inv{\alpha'}{a'}}{\beta'}{\gamma'} = 0$. This means that $\omega(\inv{\alpha'}{a'}, \beta',\gamma')$ satisfies \texttt{[.*d00*]}. In other words, $\omega(\alpha', \beta', \gamma')$ satisfies \texttt{[.*d00*]} (if $a' = 0$) or \texttt{[.*e44*]} (if $a' = 1$). Let us consider these two subcases.

    \textbf{Case 4.1.} $\omega(\alpha', \beta', \gamma')$ satisfies \texttt{[.*d00*]}, i.e. $a' = 0$. Thus, $x = \padp{1}{1}{\alpha}{\beta}{\inv{\gamma}{a}} = 0$ and $\omega(\alpha, \beta, \inv{\gamma}{a})$ must satisfy some pattern provided by Theorem~\ref{th:padpzeros}: $\overline{\gamma}$-symmetric to itself \texttt{[$\mathtt{\alpha_0}$*]} and \texttt{[$\mathtt{\beta_0}$*]} (1.1 and 1.2); $\overline{\gamma}$-symmetric \texttt{[0$\widehat{\mathtt{0}}$*1.*]} and \texttt{[1$\widehat{\mathtt{1}}$*0.*]}, 
    \texttt{[0$\widehat{\mathtt{0}}$*2$\mathtt{\alpha_0}$*]} and \texttt{[1$\widehat{\mathtt{1}}$*3$\mathtt{\alpha_0}$*]}, \texttt{[0$\widehat{\mathtt{0}}$*4$\mathtt{\beta_0}$*]} and \texttt{[1$\widehat{\mathtt{1}}$*5$\mathtt{\beta_0}$*]} (1.4, 1.5, 1.6, 1.10, 1.11 and 1.12); the two rest \texttt{[0$\widehat{\mathtt{0}}$*]} and   \texttt{[1$\widehat{\mathtt{1}}$*6$\mathtt{\gamma_0}$*]}.
    
    The pattern \texttt{[0$\widehat{\mathtt{0}}$*]} satisfies \texttt{[.*e]}, it and its $\overline{\gamma}$-symmetric pattern correspond to 1.3 and 1.9. Finally, we divide \texttt{[1$\widehat{\mathtt{1}}$*6$\mathtt{\gamma_0}$*]} into two patterns: \texttt{[1$\widehat{\mathtt{1}}$*6]} and \texttt{[1$\widehat{\mathtt{1}}$*6$\mathtt{\gamma_0}$*$\mathtt{0_6}$]}. They are 1.13 and 1.14, their $\overline{\gamma}$-symmetric patterns are 1.7 and 1.8. The first column of Table~\ref{table:zerosOfXR} is completed.
  
    \textbf{Case 4.2.} $\omega(\alpha', \beta', \gamma')$ satisfies \texttt{[.*e44*]}, i.e. $a' = 1$. Thus, $x = \padp{0}{1}{\alpha}{\beta}{\inv{\gamma}{a}} = 0$. This case is symmetric to the already considered case 3.2 relatively swapping $\alpha$ and $\beta$. Hence, we just need to choose $\alpha\beta$-symmetric patterns to the ones described in the case 3.2. The third column of Table~\ref{table:zerosOfXR} is completed.

    The theorem is proved.
\end{proof}

\subsection{Estimations for the number of all impossible differentials}\label{sec:impossibledifferentialsest}

Let $\nzer{n}{r}$ be the number of distinct $(\alpha, \beta, \gamma) \in \zspc{n} \times \zspc{n} \times \zspc{n}$ such that $\adpxr{\alpha}{\beta}{\gamma}{r} = 0$.
\begin{corollary}\label{cor:adpXRzeros_estimations}
 The following estimations hold.
    \begin{enumerate}
        \item $\nzer{n}{r} \leq (\frac{1}{7} + \frac{1}{2^{r + 1}}) 8^n - \frac{1}{7}8^r + \frac{4}{7}(8^{r - 1} - 1) (\frac{1}{5} 8^{n - r} + 7 \cdot 4^{n - r} - \frac{124}{5}3^{n-r-2})$, where $2 \leq r \leq n - 2$;
        \item $\nzer{n}{n-1} \leq (\frac{9}{28} + \frac{1}{2^n}) 8^n -\frac{88}{7}$;
        \item $\nzer{n}{r} \geq \frac{1}{7} 8^n - \frac{1}{7} 8^r$ for any $1 \leq r \leq n - 1$;
        \item $\nzer{n}{1} = \frac{5}{14} 8^{n} - \frac{6}{7}$ for any $n \geq 2$.
    \end{enumerate}
\end{corollary}
\begin{proof}
    Since each pattern presented in Table~\ref{table:zerosOfXR} is of fixed length ($r$ for the column marks and $n - r$ for the elements) and  contains at most two \texttt{*}, we just only need the following formula:
    $$
        \frac{a^{k + 1} - b^{k + 1}}{a - b} = a^k b^0 + a^{k-1}b^1 + \ldots + a^{0}b^{k}, \text{ where } a, b \in \mathbb{Z}.
    $$
    Table~\ref{table:adpXRzero_classification} classifies all patterns for words of length $n-r$. Also, it provides the number of patterns that have the same $\sigma_{n-r}$. This is the number of words of length $n - r$ that satisfy the pattern which is calculated using the formula above. 
    Let us consider the eighth pattern. The number $\sigma_k($\texttt{[.*d00*]}$)$ of all words of length $k$ (we are interested in $k = n-r$ or $k = r$) satisfying the pattern is $4 \cdot (8^{k-2} \cdot 1 + 8^{k - 3} \cdot 1 + \ldots + 8^0 \cdot 1) = \frac{4}{7} (8^{k - 1} - 1)$. We have the restriction $k \geq 2$, otherwise there are no such words at all. However, $\frac{4}{7} (8^{k - 1} - 1) = 0$ for $k = 1$. Thus, the formula for $\sigma_k($\texttt{[.*d00*]}$)$ is correct for any $k \geq 1$, i.e. in fact the restriction is $k \geq 1$. Similar patterns for words of length $n - r$ are presented in Table~\ref{table:zerosOfXR} (column~$6$) two times. We combine them since $\sigma_{n-r}($\texttt{[.*d00*]}$) = \sigma_{n-r}($\texttt{[.*e11*]}$)$. 
    Also, the marks of columns~$1$--$4$ have the same type, but they are for words of length $r$.
    The other parameters presented in Table~\ref{table:adpXRzero_classification} are obtained in the same way. 
    
\begin{table}
\begin{center}
\begin{tabular}{@{}cllll@{}}
	\toprule
         &  Pattern & \# of the same/similar & $\sigma_{n - r}$  & Restriction \\ 
	\midrule
    1 & \texttt{[$\alpha_0$*]}  & $4$ (columns $1-4$) & $4^{n-r}$ & $n - r \geq 1$ \\
    2 & \texttt{[0$\widehat{\mathtt{0}}$*]}  & $6$ (columns $1-4$) & $3^{n - r - 1}$ & $n - r \geq 1$ \\
    3 & \texttt{[0$\widehat{\mathtt{0}}$*7]}  & $2$ (columns $1-4$) & $3^{n - r - 2}$ & $n - r \geq 2$ \\
    4 & \texttt{[0$\widehat{\mathtt{0}}$*1.*]}  & $8$ (columns $1-4$) & $\frac{1}{5} (8^{n - r - 1} - 3^{n - r - 1})$ & $n - r \geq 1$ \\
    5 & \texttt{[0$\widehat{\mathtt{0}}$*2$\alpha_0$*]}  & $8$ (columns $1-4$) & $4^{n - r - 1} - 3^{n - r - 1}$ & $n - r \geq 1$ \\
    6 & \texttt{[0$\widehat{\mathtt{0}}$*7$\gamma_1$*$\mathtt{1_7}$]}  & $8$ (columns $1-4$) & $2 (4^{n - r - 2} - 3^{n - r - 2})$ & $n - r \geq 2$ \\
    7 & \texttt{[0*]}  & $2$ (column $5$) & $1$ & $n - r \geq 1$ \\
    8 & \texttt{[.*d00*]}  & $2$ (column $6$) & $\frac{4}{7} (8^{n - r - 1} - 1)$ & $n - r \geq 1$ \\
    9 & \texttt{[.*d]}  & $1$ (column $7$) & $\frac{1}{2} 8^{n - r}$ & $n - r \geq 1$ \\
    \bottomrule
\end{tabular}
\caption{Classification of the patterns from Table~\ref{table:zerosOfXR}}
\label{table:adpXRzero_classification}
\end{center}
\end{table}

To prove the estimations, we sum $\sigma_r(\cdot) \cdot \sigma_{n-r}(\cdot)$ for some (all) pair of patterns. For columns $5$--$7$ and natural $n - r \geq 1$ we obtain the following sum:
$$
    \frac{1}{2}8^r \cdot 2 + 8^{r} \cdot 2 \cdot \frac{4}{7} (8^{n - r - 1} - 1) + 4^{r} \cdot \frac{1}{2}8^{n - r} = (\frac{1}{7} + \frac{1}{2^{r + 1}}) 8^n - \frac{1}{7}8^r.
$$

The marks of columns~$1$--$4$ of Table~\ref{table:zerosOfXR} have the same $\sigma_r(\cdot) = \frac{4}{7}(8^{r - 1} - 1)$. Let us find the sum of $\sigma_{n-r}(\cdot)$ for all patterns presented in these columns. If $n - r \geq 2$, we get the following number: 
\begin{multline*}
    4 \cdot 4^{n - r} + (6 \cdot 3 + 2) \cdot 3^{n - r - 2} + 8 \cdot (\frac {1}{5}(8^{n - r - 1} - 3^{n - r - 1}) + 4^{n - r - 1} - 3^{n - r - 1} + \\ 2 \cdot (4^{n - r - 2} - 3^{n - r - 2})) =
    \frac{1}{5} 8^{n - r} + 7 \cdot 4^{n - r} - \frac{124}{5}3^{n-r-2}.
\end{multline*}
The first point is proved.

Let us estimate $\nzer{n}{n-1}$. We need to consider only patterns with the restriction $n - r \geq 1$. Columns~$1$-$4$ generate at most $\frac{4}{7}(8^{n - 2} - 1) \cdot (4 \cdot 4 + 6 \cdot 3^0) = \frac{11}{7 \cdot 8} 8^n - \frac{88}{7}$ words. Column~$5$ and Column~$7$ provide $\frac{1}{2}8^{n-1} \cdot 2 \cdot 1 = \frac{1}{8} 8^n$ and $4^{n-1} \cdot 1 \cdot 4 = 4^n$ words respectively. Column~$6$ does not provide any word. Thus,
$$
    \nzer{n}{n-1} \leq (\frac{11 + 7}{7 \cdot 8} + \frac{1}{2^n}) 8^n - \frac{88}{7} = (\frac{9}{28} + \frac{1}{2^n}) 8^n -\frac{88}{7}.
$$

Let us prove the lower bound. We note that there is no word satisfying both a pattern from column~$5$ and a pattern from column~$6$. Thus, $\nzer{n}{r} \geq  \frac{1}{2}8^r \cdot 2 + 8^{r} \cdot 2 \cdot \frac{4}{7} (8^{n - r - 1} - 1) = \frac{1}{7} 8^n - \frac{1}{7} 8^r$.

Let us calculate $\nzer{n}{1}$. The lower bound gives us the exact number of words satisfying the patterns of columns $5$ or $6$ since their $(n-r)$-parts do not intersect. The intersections of the marks are \texttt{[$\mathtt{2_4}$]} (for  columns $5$ and $7$) and \texttt{[$\gamma_0$]} (for columns $6$ and $7$). 
Also, the intersection of \texttt{[1*]} and \texttt{[.*d]} is \texttt{[1*]} (columns $5$ and $7$). The intersection of \texttt{[*e11*]} and \texttt{[.*d]} is \texttt{[*e11*]} (columns $6$ and $7$). There are no other intersections. Thus,
$$
    \nzer{n}{1} = \frac{1}{7} 8^n - \frac{1}{7} 8 + 4 \cdot \frac{1}{2}8^{n - 1} - 2 \cdot 1 - 4 \cdot \frac{4}{7}(8^{n - 2} - 1) = \frac{5}{14}8^{n} - \frac{6}{7}.
$$
The corollary is proved. 
\end{proof}
\begin{corollary}\label{cor:adpXRzeros_estimations_r_1}
Let $n \geq 5$. Then $\nzer{n}{r} < \nzer{n}{1}$ for any $r$ satisfying $2 \leq r \leq  n - 1$.
\end{corollary}
\begin{proof}
    We rewrite the upper bound for $\nzer{n}{r}$ obtained in Corollary~\ref{cor:adpXRzeros_estimations}. First of all, $\nzer{n}{1} = \frac{5}{14} 8^n - \frac{6}{7}$.
    Let $2 \leq r \leq n - 2$, i.e. $r + 1, n-r+1 \geq 3$. Also, $n \geq 5$, i.e. $(r + 1) + (n - r + 1) \geq 7$. Removing all ``$-$'' from the estimation for $\nzer{n}{r}$, we get
    \begin{multline*}
      \nzer{n}{r} + 1 \leq (\frac{1}{7} + \frac{1}{2^{r + 1}}) 8^n + \frac{4}{7}8^{r - 1} (\frac{1}{5} 8^{n - r} + 7 \cdot 4^{n - r}) = \frac{8^n}{7} + \frac{8^n}{70} + \frac{8^n}{2^{r + 1}} + \frac{8^n}{2^{n - r + 1}} \\
      \leq (\frac{1}{7} + \frac{1}{70} + \frac{1}{8} + \frac{1}{16}) 8^n = \frac{80 + 8 + 70 + 35}{5 \cdot 8 \cdot 14} 8^n < \frac{5}{14} 8^n.
    \end{multline*}
    Corollary~\ref{cor:adpXRzeros_estimations} provides that $\nzer{n}{n-1} < \nzer{n}{1}$ as well since $(\frac{9}{28} + \frac{1}{2^n}) 8^n < \frac{5}{14} 8^n$ for any $n \geq 5$. The statement is proved. 
\end{proof}
Computational experiments show us that $\nzer{n}{r} < \nzer{n}{1}$ for any $1 < r < n \leq 4$ as well, see Table~\ref{table:adpXRzeroNb}. 
\begin{table}
\begin{center}
\begin{tabular}{@{}cll|cll|cll@{}}
	\toprule
       $n$  & $r$ & $\nzer{n}{r}$ & $n$  & $r$ & $\nzer{n}{r}$ & $n$  & $r$ & $\nzer{n}{r}$ \\ 
	\midrule
    2 & 1 & 22 & 3 & 2 & 150 & 4 & 2 & 1166\\
    3 & 1 & 182 & 4 & 1 & 1462 & 4 & 3 & 1046\\
    \bottomrule
\end{tabular}
\caption{$\nzer{n}{r}$ for $n \leq 4$}
\label{table:adpXRzeroNb}
\end{center}
\end{table}
Thus, the most significant estimations are
$$
    \frac{1}{8} 8^n \leq \frac{1}{7} 8^n - \frac{1}{7} 8^r \leq \nzer{n}{r} < \lfloor \frac{5}{14} 8^n \rfloor = \nzer{n}{1} \text{ for any } r \text{ and } n,\ 2 \leq r \leq n - 1.
$$
At the same time, it was proved~\cite{LipmaaEtAl2004} that the number of impossible differentials of the function $x \oplus y$ is equal to $\frac{4}{7}(8^n - 1) = \lfloor \frac{4}{7} 8^n \rfloor$. Thus, even the one bit left rotation noticeably reduces the number of impossible differentials.

\section{Conclusion}

By rewriting the formula from~\cite{VelichkovEtAl2011}, we have obtained some argument symmetries of $\adpxrsym$. They turn out to be similar to the symmetries of $\adpsym$~\cite{MouhaEtAl2021}. Also, if the rotation is one bit left/right, we have found maximums of $\adpxrsym$, where one of its input differences is fixed. 
Though these rotations are not often used, the results obtained show us that the optimal differentials of the XR-operation may be theoretically found in some cases. 
The proposed auxiliary notions also turn out to be sufficient for finding all impossible differentials of $(x \oplus y) \lll r$.

Note that anything that is true for $\adpxrsym$ works for $\adprxsym$ as well (taking into account positions of arguments and the rotation parameter). The maximums for other argument fixations as well as for other rotations are the topics for future research.

\bigskip

\noindent {\normalsize \textbf{Acknowledgements} 
The work is supported by the Mathematical Center in Akademgorodok under the agreement No. 075--15--2022--282 with the Ministry of Science and Higher Education of the Russian Federation. The authors are very grateful to Nicky Mouha.}

\begin{appendices}

\section{Proof of Theorem~\ref{th:padpzeros}}\label{sec:proofofpadpzeros}

First of all, we prove some auxiliary lemmas. 
\begin{lemma}\label{lemma:psymmetries}
    Let $\alpha, \beta, \gamma \in \zspc{n}$, $\omega = \omega(\alpha, \beta, \gamma)$ and $\pi: \zspc{3} \to \zspc{3}$ be a permutation of coordinates, i.e. $\pi(x_0, x_1, x_2) = (x_{i_0}, x_{i_1}, x_{i_2})$, where $(x_0, x_1, x_2) \in \zspc{3}$ and $\{i_0, i_1, i_2\} = \{0, 1, 2\}$. Then
    \begin{align*}
            e_{\pi(i)} A_{\pi(\omega_{0})} \ldots A_{\pi(\omega_{n - 1})} e^T_0 = e_{i} A_{\omega_{0}} \ldots A_{\omega_{n - 1}} e^T_0.
    \end{align*}
    We consider $\{0, \ldots, 7\}$ as the integer representations for elements of $\zspc{3}$.
\end{lemma}
\begin{proof}
    Let $S_{\pi}: \mathbb{Q}^8 \to \mathbb{Q}^8$ such that 
    $$
        S(x_0, \ldots, x_7) = (x_{\pi^{-1}(0)}, \ldots x_{\pi^{-1}(7)}) \text{ for any } x \in \mathbb{Q}^8.
    $$
    We can consider $S_{\pi}$ as a binary matrix of size $8 \times 8$ since it is a linear function over $\mathbb{Q}$.
    Let us show that $S_{\pi} A_k S_{\pi}^{-1} = A_{\pi(k)}$.
    It is not difficult to see that for any binary matrices $B$ and $B'$ of size $8\times8$ the equality $B' = S_{\pi} B S_{\pi}^{-1}$ means that $B'_{i, j} = B_{\pi^{-1}(i), \pi^{-1}(j)}$.

    We can easily check that $(A_0)_{i, j} \neq 0$ $\Longleftrightarrow$ $\wt{j}$ is even and $j \preceq i$. Also we know that $(A_k)_{i, j} = (A_0)_{i \oplus k, j \oplus k}$. Thus, 
    \begin{align*}
        (S_{\pi} A_k S_{\pi}^{-1})_{i, j} \neq 0 &\Longleftrightarrow (A_k)_{\pi^{-1}(i), \pi^{-1}(j)} \neq 0 \\ &\Longleftrightarrow \wt{\pi^{-1}(j) \oplus k} \text{ is even and } \pi^{-1}(j) \oplus k \preceq \pi^{-1}(i) \oplus k\\
        &\Longleftrightarrow \wt{j \oplus \pi(k)} \text{ is even and } j \oplus \pi(k) \preceq i \oplus \pi(k)\\
        &\Longleftrightarrow (A_{\pi(k)})_{i, j} \neq 0.
    \end{align*}
    Matrices $A_k$ contain only $0, 1$ and $\frac{1}{4}$. Moreover, $(A_{k})_{i, j} = 1$ $\Longleftrightarrow$ $i = j = k$, which is consistent with
    $$
        (A_{\pi(k)})_{\pi(k), \pi(k)} = 1 = (A_k)_{\pi^{-1}(\pi(k)), \pi^{-1}(\pi(k))} = (S_{\pi} A_k S_{\pi}^{-1})_{\pi(k), \pi(k)}.
    $$
    It means that $S_{\pi} A_k S_{\pi}^{-1} = A_{\pi(k)}$. 
    Hence,
    \begin{align*}
        e_{i} A_{\omega_{0}} \ldots A_{\omega_{n - 1}} e^T_0 &= e_{i} (S_{\pi}^{-1} S_{\pi}) A_{\omega_{0}}  (S_{\pi}^{-1} S_{\pi})  \ldots  (S_{\pi}^{-1} S_{\pi}) A_{\omega_{n - 1}} (S_{\pi}^{-1} S_{\pi}) e^T_0 \\
        &= (e_{i} S_{\pi}^{-1}) (S_{\pi} A_{\omega_{0}}  S_{\pi}^{-1}) \ldots (S_{\pi} A_{\omega_{n - 1}} S_{\pi}^{-1}) (S_{\pi} e^T_0)\\
        &= e_{\pi(i)} A_{\pi(\omega_{0})} \ldots A_{\pi(\omega_{n - 1})} e^T_0,
    \end{align*}
    since $e_{i} S_{\pi}^{-1} = e_{\pi(i)}$ and $S_{\pi} e^T_0 = e^T_0$.
\end{proof}
\begin{remark}\label{remark:psymmetriesusage}
    We will use Lemma~\ref{lemma:psymmetries} to unite symmetric cases. For instance, let us suppose that for any $\alpha, \beta, \gamma \in \zspc{n}$ the following holds: $
    (e_3 + e_5) A_{\omega_{0}} \ldots A_{\omega_{n - 1}} e_0^T = 0 \iff \gamma = 0, \text{ where } \omega = \omega(\alpha, \beta, \gamma).$ Next, we choose the permutation $\pi: \zspc{3} \to \zspc{3}$ that swaps $\beta$ and $\gamma$, i.e. $\pi(a, b, c) = (a, c, b)$ for all $a,b,c \in \zspc{}$. Lemma~\ref{lemma:psymmetries} gives us that 
    $$
        (e_3 + e_5) A_{\pi(\omega_{0})} \ldots A_{\pi(\omega_{n - 1})} e_0^T = (e_{\pi(3)} + e_{\pi(5)}) A_{\pi(\pi(\omega_{0}))} \ldots A_{\pi(\pi(\omega_{n - 1}))} e_0^T.
    $$
    Thus, $(e_{3} + e_{6}) A_{\omega_{0}} \ldots A_{\omega_{n - 1}} e_0^T = 0 \iff \beta = 0$ for any $\alpha, \beta, \gamma \in \zspc{n}$. Indeed, $\pi(3) = 3$, $\pi(5) = 6$, $\pi(\pi(\omega_i)) = \omega_i$ for $i \in \{0, \ldots, n - 1\}$ and $(e_3 + e_5) A_{\pi(\omega_{0})} \ldots A_{\pi(\omega_{n - 1})} e_0^T = 0$ $\iff$ $\gamma = 0$ by the assumption.
\end{remark}
\begin{table}
\begin{center}
\begin{minipage}{\linewidth}
\begin{tabular*}{\linewidth}{@{\extracolsep{\fill}}ccccccccc@{\extracolsep{\fill}}}
	\toprule
       $j \setminus i$     & $0$ & $1$ & $2$ & $3$ & $4$ & $5$ & $6$ & $7$ \\
	\midrule
        $0$ & $-$ & $v_2 + v_4$ & $v_1 + v_4$ & $-$ & $v_1 + v_2$ & $-$ & $-$ & $0$ \\
        $1$ & $v_3 + v_5$ & $-$ & $-$ & $v_0 + v_5$ & $-$ & $v_0 + v_3$ & $0$ & $-$ \\
        $2$ & $v_3 + v_6$ & $-$ & $-$ & $v_0 + v_6$ & $-$ & $0$ & $v_0 + v_3$ & $-$ \\
        $3$ & $-$ & $v_2 + v_7$ & $v_1 + v_7$ & $-$ & $0$ & $-$ & $-$ & $v_1 + v_2$ \\
        $4$ & $v_5 + v_6$ & $-$ & $-$ & $0$ & $-$ & $v_0 + v_6$ & $v_0 + v_5$ & $-$ \\
        $5$ & $-$ & $v_4 + v_7$ & $0$ & $-$ & $v_1 + v_7$ & $-$ & $-$ & $v_1 + v_4$ \\
        $6$ & $-$ & $0$ & $v_4 + v_7$ & $-$ & $v_2 + v_7$ & $-$ & $-$ & $v_2 + v_4$ \\
        $7$ & $0$ & $-$ & $-$ & $v_5 + v_6$ & $-$ & $v_3 + v_6$ & $v_3 + v_5$ & $-$ \\
	\bottomrule
\end{tabular*}
\end{minipage}
\end{center}
\caption{The $j$-th coordinate of $4 \cdot A_i v$, i.e. $4 \cdot e_j A_i v$, where $v^T \in \mathbb{Q}^8$}
\label{table:nexttwocoords}
\end{table}
\begin{lemma}\label{lemma:akv}
    Let $v^T \in \mathbb{Q}^8$ and $i \in \{0, \ldots, 7\}$. Then $j$-th coordinates of $4 \cdot A_i v$, where the parities of $\wt{i}$ and $\wt{j}$ are different, are presented in Table~\ref{table:nexttwocoords}.
\end{lemma}
The proof is straightforward (the matrices $A_0, \ldots, A_7$ can be found in Section~\ref{sec:preliminaries}).

\begin{lemma}\label{lemma:sumtwozero}
    Let $\omega = \omega(\alpha, \beta, \gamma)$, $\alpha, \beta, \gamma \in \zspc{n}$ such that $(\alpha, \beta, \gamma) \neq (0,0,0)$.
    Let $v = A_{\omega_{0}} \ldots A_{\omega_{n - 1}} e^T_0$ and $v_0 + v_1 + \ldots + v_7 \neq 0$, i.e. $\adp{\alpha}{\beta}{\gamma} \neq 0$. Then for the sums from Table~\ref{table:nexttwocoords} the following holds:
    \begin{enumerate}
        \item $v_3 + v_5 = 0$ $\Longleftrightarrow$ $v_1 + v_7 = 0$ $\Longleftrightarrow$ $\gamma = 0$;
        \item $v_3 + v_6 = 0$ $\Longleftrightarrow$ $v_2 + v_7 = 0$ $\Longleftrightarrow$ $\beta = 0$;
        \item $v_5 + v_6 = 0$ $\Longleftrightarrow$ $v_4 + v_7 = 0$ $\Longleftrightarrow$ $\alpha = 0$;
        \item all $v_0 + v_3$, $v_0 + v_5$ and $v_0 + v_6$ are not zero;
        \item all $v_1 + v_2$, $v_1 + v_4$ and $v_2 + v_4$ are not zero.
    \end{enumerate}
    Also, $v_0 = 1$ and $v_1 = v_2 = \ldots = v_7 = 0$ if $(\alpha, \beta, \gamma) = (0, 0, 0)$.
\end{lemma}
\begin{proof}
    \textbf{Points 1, 2, 3.} 
    Let us consider 
    $$
        L_1 A_0 v = \frac{1}{4}(v_3 + v_5) + \frac{1}{4} v_3 + \frac{1}{4}v_5 + 0 = \frac{1}{2}v_3 + \frac{1}{2}v_5.
    $$
    By definition, $\cadp{1}{\alpha'}{\beta'}{\gamma'} = L_1 A_0 v$, where $\alpha' = \cctn{0}{\alpha}, \beta' = \cctn{0}{\beta}, \gamma' = \cctn{0}{\gamma}$.
    At the same time, $\adp{\alpha'}{\beta'}{\gamma'} \neq 0$ since $\adp{\alpha}{\beta}{\gamma} \neq 0$ and $(\alpha,\beta, \gamma) \neq (0, 0, 0)$, see Proposition~\ref{prop:adpzeropattern}. Thus, Proposition~\ref{prop:cadpzeros} guaranties that $\cadp{1}{\alpha'}{\beta'}{\gamma'} \neq 0$ if and only if $\gamma' \neq 0$, which is equivalent to the following: $v_3 + v_5 \neq 0$ if and only if $\gamma \neq 0$.
    Similarly,
    $$
        L_1 A_2 v = \frac{1}{4} v_1 + \frac{1}{4}(v_1 + v_7) + 0  + \frac{1}{4}v_7 = \frac{1}{2}v_1 + \frac{1}{2}v_7
    $$
    and $\cadp{1}{\alpha''}{\beta''}{\gamma''} = L_1 A_2 v$, where $\alpha'' = \cctn{0}{\alpha}, \beta'' = \cctn{1}{\beta}, \gamma'' = {0}\cctn{\gamma}$. The reasons above guarantee that $v_1 + v_7 \neq 0$ if and only if $\gamma \neq 0$.

    The rest points can be proved by symmetries. Let $\pi_{\beta \gamma}, \pi_{\alpha \gamma}: \zspc{3} \to \zspc{3}$ such that $\pi_{\beta \gamma}(a, b, c) = (a, c, b)$ and $\pi_{\alpha \gamma}(a, b, c) = (c, b, a)$, $a,b,c \in \zspc{}$. Since $v_3 + v_5 = (e_3 + e_5) v$, Lemma~\ref{lemma:psymmetries} and Remark~\ref{remark:psymmetriesusage} give us that 
    $(e_{\pi_{\beta \gamma}(3)} + e_{\pi_{\beta \gamma}(5)}) v = (e_{3} + e_{6}) v = 0 \iff \beta = 0$. Similarly, $\pi_{\beta \gamma}(1) = 2$ and $\pi_{\beta \gamma}(7) = 7$, $\pi_{\alpha \gamma}(3) = 6$ and $\pi_{\alpha \gamma}(5) = 5$, $\pi_{\alpha \gamma}(1) = 4$ and $\pi_{\alpha \gamma}(7) = 7$.

    \textbf{Points 4, 5.} Let us consider 
    \begin{align*}
        L_0 A_3 v &= \frac{1}{4} v_0 + \frac{1}{4}(v_0 + v_6) + 0 + \frac{1}{4}v_6 = \frac{1}{2}v_0 + \frac{1}{2}v_6.
    \end{align*}
    Proposition~\ref{prop:cadpzeros} guaranties that $L_0 A_3 v = \cadp{0}{\alpha'}{\beta'}{\gamma'} \neq 0$, where $\alpha' = \cctn{0}{\alpha}, \beta' = \cctn{1}{\beta}, \gamma' = \cctn{1}{\gamma}$.
    Next, 
    \begin{align*}
        L_0 A_1 v &= \frac{1}{4}(v_2 + v_4) + \frac{1}{4} v_2 + \frac{1}{4}v_4 + 0 = \frac{1}{2}v_2 + \frac{1}{2}v_4.
    \end{align*}
    Proposition~\ref{prop:cadpzeros} guaranties that $L_0 A_1 v = \cadp{0}{\alpha'}{\beta'}{\gamma'} \neq 0$, where $\alpha' = \cctn{0}{\alpha}, \beta' = \cctn{0}{\beta}, \gamma' = \cctn{1}{\gamma}$. 
    Similarly to the previous points, $\pi_{\beta \gamma}$, $\pi_{\alpha \gamma}$, Lemma~\ref{lemma:psymmetries} and Remark~\ref{remark:psymmetriesusage} give us all the rest.
\end{proof}

The following lemma shows us that the condition for $\padp{a}{b}{\alpha}{\beta}{\gamma} = 0$ depends on the most significant bits of $\alpha$ and $\beta$. 
\begin{lemma}\label{lemma:padpzeros}
    Let $\alpha, \beta, \gamma \in \zspc{n}$, $a, b \in \zspc{}$, $\adp{\alpha}{\beta}{\gamma} \neq 0$ and $(a, b) \neq (\overline{\alpha_{0}}, \overline{\beta_{0}})$. Then $\padp{a}{b}{\alpha}{\beta}{\gamma} = 0$ $\iff$ any of the following two conditions holds: $a = 1$ and $\alpha = 0$; $b = 1$ and $\beta = 0$.
\end{lemma}
\begin{proof}
    By Theorem~\ref{th:cadpsymmetries}, $\padp{1}{b}{0}{\beta}{\gamma} = 0$ and $\padp{a}{1}{\alpha}{0}{\gamma} = 0$. Hence, we need to show that it is not zero in all other cases.

    Let $\omega = \omega(\alpha, \beta, \gamma)$, $v = A_{\omega_{1}} \ldots A_{\omega_{n - 1}}e^T_0$, $k = \omega_{0}$ and $T_k$ be the involution matrix that swaps $i$ and $i \oplus k$ coordinates of $x$, $x \in \mathbb{Q}^8$, $i \in \{0, \ldots, 7\}$.
    Then 
    \begin{equation}\label{eq:padpmsb}
        \padp{a}{b}{\alpha}{\beta}{\gamma} = L_{a,b} A_k v = L_{a,b} T_k A_0 T_k v = L_{a \oplus k_2, b \oplus k_1} A_0 (T_k v).
    \end{equation}
    
    Since $(a, b) \neq (\overline{k_2}, \overline{k_1})$, we need to consider $(a \oplus k_2, b \oplus k_1) \in \{ (0, 0), (0, 1), (1, 0) \}$. 
    
    \textbf{Case 1.} $(a \oplus k_2, b \oplus k_1) = (0, 0)$. According to (\ref{eq:padpmsb}),
    $$
         \padp{a}{b}{\alpha}{\beta}{\gamma} = L_{0,0} A_0 (T_k v) = v_k + \frac{1}{2}v_{k \oplus 3} + \frac{1}{2}v_{k \oplus 5} + \frac{1}{4}v_{k \oplus 6},
    $$
    that is zero $\iff$ $v_0 + v_3 + v_5 + v_6 = 0$ (for $k \in \{0, 3, 5, 6\}$) and  $v_1 + v_2 + v_4 + v_7 = 0$ (for $k \in \{1, 2, 4, 7\}$). 

    If $n > 1$ and $(\omega_{1}, \ldots, \omega_{n-1}) \neq (0, \ldots, 0)$, they are both not zero due to Lemma~\ref{lemma:sumtwozero}.
    Otherwise, $v = e^T_0$ and $k \in \{0, 3, 5, 6\}$, where $k \in \{0, 3, 5, 6\}$ (due to $\adp{\alpha}{\beta}{\gamma} \neq 0$ and Proposition~\ref{prop:adpzeropattern}). It means that $v_0 + v_3 + v_5 + v_6$ is not zero.

    \textbf{Case 2.} $(a \oplus k_2, b \oplus k_1) = (0, 1)$. According to (\ref{eq:padpmsb}),
    $$
         \padp{a}{b}{\alpha}{\beta}{\gamma} = L_{0,1} A_0 (T_k v) = \frac{1}{2}v_{k \oplus 3} + \frac{1}{4}v_{k \oplus 6},
    $$
    that is zero $\iff$ the sum $v_{k \oplus 3} + v_{k \oplus 6}$ is zero. 
    
    Let $k = 2, 3, 6, 7$, i.e. $\beta_{0} = 1$ and $\beta \neq 0$. This sum is equal to $v_1 + v_4$, $v_0 + v_5$, $v_0 + v_5$ and $v_1 + v_4$ respectively. If $n > 1$ and $(\omega_{1}, \ldots, \omega_{n-1}) \neq (0, \ldots, 0)$, all of them are not zero due to Lemma~\ref{lemma:sumtwozero}. Otherwise, $v = e^T_0$ and $k \in \{0, 3, 5, 6\}$, i.e. $k = 3, 6$. This means that $v_0 + v_5$ is not zero.

    Let $k = 0,1,4,5$, i.e. $\beta_{0} = 0$. Then the sum is equal to  $v_3 + v_6$, $v_2 + v_7$, $v_2 + v_7$ and $v_3 + v_6$ respectively. If $n > 1$ and $(\omega_{1}, \ldots, \omega_{n-1}) \neq (0, \ldots, 0)$, all of them are not zero $\iff$ $\beta \neq 0$. Indeed, $\beta_{0} = 0$ due to the values of $k$. Together with Lemma~\ref{lemma:sumtwozero}, this transforms to $\beta \neq 0$.
    Otherwise, $v = e^T_0$ (therefore, $\beta = 0$) and $k \in \{0, 3, 5, 6\}$, i.e. $k = 0, 5$. In this case $v_3 + v_6 = 0$. 

    At the same time, $k = 0, 1$ corresponds to $(a,b) = (0, 1)$ and $k = 4,5$ corresponds $(a,b) = (1,1)$. This means that the sum is zero only if $b = 1$ and $\beta = 0$.

    \textbf{Case 3.} $(a \oplus k_2, b \oplus k_1) = (1, 0)$. According to (\ref{eq:padpmsb}),
    $
         \padp{a}{b}{\alpha}{\beta}{\gamma} = \frac{1}{2}v_{k \oplus 5} + \frac{1}{4}v_{k \oplus 6}.
    $
    This is symmetric to the previous case by Lemma~\ref{lemma:psymmetries} and Remark~\ref{remark:psymmetriesusage} since we can choose $\pi$ that swaps the bits of $\alpha$ and $\beta$ and obtain that $\frac{1}{2}v_{k \oplus 5} + \frac{1}{4}v_{k \oplus 6}$ and $\frac{1}{2}v_{\pi(k) \oplus 3} + \frac{1}{4}v_{\pi(k) \oplus 6}$ have symmetric properties relatively swapping the conditions for $\alpha$ and $\beta$.
\end{proof}

Finally, we prove Theorem~\ref{th:padpzeros}.
\begin{proof}
    Let $\omega = \omega(\alpha, \beta, \gamma)$. 
    First of all, $\padp{a}{b}{\alpha}{\beta}{\gamma} = 0$ if $\adp{\alpha}{\beta}{\gamma} = 0$ by definition and $\adp{\alpha}{\beta}{\gamma} = 0 \iff \omega$ satisfies {\rm\texttt{[.*d0*]}} (see Proposition~\ref{prop:adpzeropattern}).
    It corresponds to the first row of Table~\ref{table:padpzeros}.
    
    Next, we assume that $\adp{\alpha}{\beta}{\gamma} \neq 0$.
    Lemma~\ref{lemma:padpzeros} guaranties that $\padp{a}{b}{\alpha}{\beta}{\gamma} = 0$ if $\omega(\alpha, \beta, \gamma)$ satisfies \texttt{[$\mathtt{\alpha_0}$*]} and \texttt{[$\mathtt{\beta_0}$*]} for $(a, b) \in \{(1, 0), (1, 1)\}$ and $(a, b) \in \{(0, 1), (1, 1)\}$ respectively. 
    Also, it provides that $(a, b) = (\overline{\alpha_{0}}, \overline{\beta_{0}})$ is a necessary condition for $\padp{a}{b}{\alpha}{\beta}{\gamma} = 0$ in all other cases. In addition, we assume that $\omega_{0}
    = (\overline{a}, \overline{b}, c)$ for the given $(a, b)$ and some $c \in \zspc{}$.
   
    Let $t = \omega_{0}$, $v = A_{\omega_{1}} \ldots A_{\omega_{n - 1}} e_0^T$, $\alpha', \beta', \gamma'$ be $\alpha, \beta, \gamma$ without their most significant bits and $\omega' = \omega(\alpha', \beta', \gamma')$. Then by definition
    \begin{align*}
        \padp{a}{b}{\cctn{\overline{a}}{\alpha'}}{\cctn{\overline{b}}{\beta'}}{\cctn{c}{\gamma'}} &= L_{a,b} A_t v = (L_{a,b} T_t) A_0 (T_t v)\\
        &= L_{a \oplus \overline{a},b \oplus \overline{b}} A_0 (v_{t\oplus0}, v_{t\oplus1}, \ldots, v_{t\oplus7})^T\\
        &= \frac{1}{4} L_{1,1}  (v_{t \oplus 0}, *,\ldots, *, v_{t \oplus 5}, v_{t\oplus6}, 0)^T\\
        &= \frac{1}{4} v_{t \oplus 6} + 0 = \frac{1}{4} v_{abc}.
    \end{align*}
    Thus, we need to determine conditions under which $v_{(a, b, c)} = 0$.
    We call the coordinate $(a, b, c) = t \oplus 6$ \textit{determining} for the given $a, b$.

    It implies that all required patterns can be described as \texttt{[t$\widehat{\mathtt{t}}$*u]}, where \texttt{u} is some subpattern for a subword $u$ of length $0 \leq m \leq n - 1$. There are three cases: $u = \varnothing$, $u_{0} = t\oplus6$ and $u_{0} \in \{0,\ldots,7\}\backslash\{t, t\oplus3, t\oplus5, t\oplus6\}$, i.e. $u_{0}$ does not satisfy both \texttt{$\widehat{\mathtt{t}}$} and \texttt{t$\oplus$6}.
    Let $v'' = A_{u_{1}} \ldots A_{u_{m - 1}} e^T_0$ if $m > 0$ and $v' = A_{u_{0}} v''$. More precisely, $v'' = e^T_0$ if $m = 1$ and $v' = e^T_0$ if $m = 0$.

    An explanation why we add \texttt{$\widehat{\mathtt{t}}$*} is the following. Let $x'^T \in \mathbb{Q}^8$, $d \in \{0, \ldots, 7\}$ satisfy \texttt{$\widehat{\mathtt{t}}$} and $x = A_{d} x'$.
    Then $x_{t\oplus6} = 0$ $\iff$ $x'_{t\oplus6} = 0$. Indeed,
    \begin{align*}
        d = t : \ \ \ \ \ \ &x_{t\oplus6} = e_{t\oplus6} A_t x' = e_{t\oplus6}T_t A_0 T_t x' = e_6 A_0 (T_t x') = 
        \frac{1}{4}x'_{t\oplus6},\\
        d = t\oplus3 : \ &x_{t\oplus6} = e_{t\oplus6} A_{t\oplus3} x' = e_{t\oplus6}T_{t \oplus 3} A_0 T_{t \oplus 3} x' = e_5 A_0 (T_{t\oplus3} x') = \frac{1}{4}x'_{t\oplus6},\\
        d = t\oplus5 : \ &x_{t\oplus6} = e_{t\oplus6} A_{t\oplus5} x' = e_{t\oplus6}T_{t \oplus 5} A_0 T_{t \oplus 5} x' =  e_3 A_0 (T_{t\oplus5} x') = \frac{1}{4}x'_{t\oplus6}.
    \end{align*}
    Here $T_t$ is the involution matrix that swaps $i$ and $i \oplus t$ coordinates, see the proof of Lemma~\ref{lemma:matrixindexXOR}.
    In other words, adding \texttt{$\widehat{\mathtt{t}}$*} to the beginning of any pattern for $u$ keeps zero (non-zero) value at the position $t\oplus6$.

    \textbf{Case 1.} $u = \varnothing$. Therefore, $v' = e^T_0$.
    Since $\adp{\alpha}{\beta}{\gamma} \neq 0$, any $y$ satisfying \texttt{$\widehat{\mathtt{t}}$} must belong to $\{0, 3, 5, 6\}$, see Proposition~\ref{prop:adpzeropattern}.
    Therefore, $t \in \{0, 3, 5, 6\}$.
    For $(a,b) \in \{(0,0), (0,1), (1,0), (1,1)\}$ we have that $t = 6, 5, 3, 0$ and the determining coordinate is $t\oplus6 = 0, 3, 5, 6$ respectively.
    Thus, if the value in the determining coordinate is zero, then \texttt{[t$\widehat{\mathtt{t}}$*]} must be included to the table. 
    At the same time, $v' = e^T_0$ has only one non-zero coordinate $v'_0 = 1$. Hence, $(a,b) = (0,0)$ gives no patterns for zero values since the determining coordinate $v'_0 = 1 \neq 0$. Also, 
    
    \begin{center}
         \texttt{[5$\widehat{\mathtt{5}}$*]} is added for $(a,b) = (0,1)$ since the determining coordinate $v'_3 = 0$,\\
         \texttt{[3$\widehat{\mathtt{3}}$*]} is added for $(a,b) = (1,0)$ since the determining coordinate $v'_5 = 0$,\\
         \texttt{[0$\widehat{\mathtt{0}}$*]} is added for $(a,b) = (1,1)$ since the determining coordinate $v'_6 = 0$.
    \end{center}
   
    \textbf{Case 2.}  $u_{0} = t \oplus 6$. Let $k = t \oplus 6$ be the determining coordinate.
    It is not difficult to see that $v'_k = v''_k + \frac{1}{4} \cdot (v''_{k\oplus3} + v''_{k\oplus5} + v''_{k\oplus6})$ and all other coordinates of $v'$ depend only on $v''_{k\oplus3}, v''_{k\oplus5}, v''_{k\oplus6}$.
    At the same time, $v' \neq (0, \ldots, 0)^T$ since $\adp{\alpha}{\beta}{\gamma} \neq 0$.
    Hence, $v'_k \neq 0$.

    \textbf{Case 3.}  $u_{0} \in \{t\oplus1, t\oplus2, t\oplus4, t\oplus7\}$, i.e. $\wt{u_{0}}$ is odd $\iff$ $\wt{t \oplus 6}$ is even. 
    Thus, Lemma~\ref{lemma:akv} and Table~\ref{table:nexttwocoords} give us all possible values of the determining coordinate $v'_{t \oplus 6}$ which determines $v_{t \oplus 6}$ and the value $\padp{a}{b}{\alpha}{\beta}{\gamma}$. More precisely, they provide the value $4 \cdot e_j A_i v''$, where $j = t \oplus 6 = (a, b, c)$ and $i = u_{0}$.
    
    Let $\alpha'', \beta''$ and $\gamma''$ be the beginnings (starting with the least significant bits) of $\alpha, \beta$ and $\gamma$ such that $\omega(\alpha'', \beta'', \gamma'') = (u_1, \ldots, u_{m-1})$. Note that only the cases of $m = 1$ and $\alpha'' = \beta'' = \gamma'' = 0$ correspond to $v'' = e_0^T$. 

    \textbf{Case 3.1.} $j = i \oplus 7$, i.e. $t \oplus 6 = u_{0} \oplus 7$, which is equivalent to $u_{0} = t \oplus 1$. Table~\ref{table:nexttwocoords} guaranties that $v'_{t \oplus 6}$ is zero. This corresponds to the patterns \texttt{[t$\widehat{\mathtt{t}}$*(t$\oplus$1).*]} for all $a,b$, where $t=(\overline{a}, \overline{b}, c)$, $c \in \zspc{}$, i.e.
    \begin{center}
        \texttt{[6$\widehat{\mathtt{6}}$*7.*]} and \texttt{[7$\widehat{\mathtt{7}}$*6.*]} for  $(a,b)=(0,0)$,\\
        \texttt{[4$\widehat{\mathtt{4}}$*5.*]} and \texttt{[5$\widehat{\mathtt{5}}$*4.*]} for  $(a,b)=(0,1)$,\\
        \texttt{[2$\widehat{\mathtt{2}}$*3.*]} and \texttt{[3$\widehat{\mathtt{3}}$*2.*]} for  $(a,b)=(1,0)$,\\
        \texttt{[0$\widehat{\mathtt{0}}$*1.*]} and \texttt{[1$\widehat{\mathtt{1}}$*0.*]} for $(a,b)=(1,1)$.
    \end{center}

    \textbf{Case 3.2.} $u_{0} \neq t \oplus 1$. Lemma~\ref{lemma:sumtwozero} allows us to describe all cases such that the determining coordinate $v'_{t \oplus 6}$ is zero.
    Also, there is no need to consider the case of $v'' = e_0^T$ separately. Indeed, points $1$--$4$ are correct for $v'' = (1, 0, \ldots, 0)^T$. Point~$5$ contains only $v_1, v_2$ and $v_4$. At the same time, $i = u_{0}$ must belong to $\{0, 3, 5, 6\}$ since $\adp{\alpha}{\beta}{\gamma} \neq 0$, see Proposition~\ref{prop:adpzeropattern}. But columns $0, 3, 5$ and $6$ of Table~\ref{table:nexttwocoords} do not contain any of $v_1, v_2$ and $v_4$. Thus, we can ignore this point.

    Next, points $4$--$5$ give us the values that are not zero. Hence, we need to consider only points $1$--$3$ and find the corresponding values in Table~\ref{table:nexttwocoords}. According to point~$1$,
    
    \begin{center}
    $(i, j) = (0, 1)$ adds \texttt{[7$\widehat{\mathtt{7}}$*0$\mathtt{\gamma_0}$*]} for $(a,b) = (0,0)$,\\
    $(i, j) = (2, 3)$ adds \texttt{[5$\widehat{\mathtt{5}}$*2$\mathtt{\gamma_0}$*]} for $(a,b) = (0,1)$,\\
    $(i, j) = (4, 5)$ adds \texttt{[3$\widehat{\mathtt{3}}$*4$\mathtt{\gamma_0}$*]} for $(a,b) = (1,0)$,\\
    $(i, j) = (6, 7)$ adds \texttt{[1$\widehat{\mathtt{1}}$*6$\mathtt{\gamma_0}$*]} for $(a,b) = (1,1)$.        
    \end{center}
    According to point~$2$,

    \begin{center}
    $(i, j) = (0, 2)$ adds \texttt{[4$\widehat{\mathtt{4}}$*0$\mathtt{\beta_0}$*]} for $(a,b) = (0,1)$,\\
    $(i, j) = (1, 3)$ adds \texttt{[5$\widehat{\mathtt{5}}$*1$\mathtt{\beta_0}$*]} for $(a,b) = (0,1)$,\\
    $(i, j) = (4, 6)$ adds \texttt{[0$\widehat{\mathtt{0}}$*4$\mathtt{\beta_0}$*]} for $(a,b) = (1,1)$,\\
    $(i, j) = (5, 7)$ adds \texttt{[1$\widehat{\mathtt{1}}$*5$\mathtt{\beta_0}$*]} for $(a,b) = (1,1)$.    
    \end{center}
    According to point~$3$,
    
    \begin{center}
    $(i, j) = (0, 4)$ adds \texttt{[2$\widehat{\mathtt{2}}$*0$\mathtt{\alpha_0}$*]} for $(a,b) = (1,0)$,\\
    $(i, j) = (1, 5)$ adds \texttt{[3$\widehat{\mathtt{3}}$*1$\mathtt{\alpha_0}$*]} for $(a,b) = (1,0)$,\\
    $(i, j) = (2, 6)$ adds \texttt{[0$\widehat{\mathtt{0}}$*2$\mathtt{\alpha_0}$*]} for $(a,b) = (1,1)$,\\
    $(i, j) = (3, 7)$ adds \texttt{[1$\widehat{\mathtt{1}}$*3$\mathtt{\alpha_0}$*]} for $(a,b) = (1,1)$.
    \end{center}
    The theorem is proved.
\end{proof}

\end{appendices}

\bibliographystyle{ieeetr}


\end{document}